\newcommand{\fullVersion}{}
\documentclass{llncs}
\ifdefined\fullVersion
\usepackage[margin=1.2in]{geometry}
\pdfoutput=1
\else
\newcommand{\TACASVersion}{}
\let\llncssubparagraph\subparagraph
\let\subparagraph\paragraph
\usepackage[compact]{titlesec}
\let\subparagraph\llncssubparagraph
\fi

\usepackage{algpseudocode}
\usepackage[Algorithm]{algorithm}
\usepackage{graphicx}
\usepackage{multirow}
\usepackage{caption}
\usepackage{amsmath}
\usepackage{amssymb}
\usepackage{subfig}
\usepackage{dblfloatfix}
\usepackage{url}
\usepackage{lipsum}
\usepackage{hhline}
\usepackage{tabularx}
\usepackage{cite}

\usepackage{amsthm}
\usepackage{enumerate}
\usepackage{array,multirow}
\allowdisplaybreaks
\usepackage{pifont}
\newcommand{\cmark}{\ding{51}}%
\usepackage{balance}
\usepackage{enumitem}
\allowdisplaybreaks
\usepackage[T1]{fontenc}
\newcommand{\xmark}{\ding{55}}%
\algtext*{EndFunction}
\algtext*{EndIf}
\algtext*{EndWhile}
\algtext*{EndFor}

\newtheorem*{lemma*}{Lemma}
\newtheorem*{theorem*}{Theorem}
\begin{document}
\title{Learning Software Constraints via Installation~Attempts
\ifdefined\fullVersion
\thanks{The research leading to these results has received funding from the European
Union under the H2020 and 5G-PPP Phase2 programs, under Grant Agreement No.
761 557 (project NGPaaS).}
\else
\thanks{The research was financially supported by the European
Union under the H2020 and 5G-PPP Phase2 programs, under Grant Agreement No.
761 557 (project NGPaaS).}
\fi
}
%
%
\author{Ran Ben Basat\inst{1}
\and
Maayan Goldstein\inst{2}
\and
Itai Segall\inst{2}
}
\authorrunning{Ben Basat et al.}
%
\institute{Harvard University \qquad{}\email{ran@seas.harvard.edu}\\ 
\and
Nokia Bell Labs\qquad{}
\email{\{maayan.goldstein,itai.segall\}@nokia.com}}
\maketitle              
\newcommand{\reachabilityMap}{\mathfrak R}
\newcommand{\GTLB}{\ensuremath{\ceil{\log_2 \sum_{i=0}^{d}{n \choose i}}}}
\newcommand{\maxProblem}{{\sc Maximal Sub-repository}}
\newcommand{\minProblem}{{\sc Minimal Package Installation}}
\newcommand{\learnAll}{{\sc Full Learning}}
\newcommand{\CFFParams}{\ensuremath{n,a,b}}
\newcommand{\CFF}{\ensuremath{\mathit{CFF}}}
\newcommand{\CFFSizeLetter}{\ensuremath{\mathcal S}}
\newcommand{\CFFSize}{\ensuremath{\CFFSizeLetter_{\CFFParams}}}
\newcommand{\CFFs}{\ensuremath{\mathit{CFFs}}}
\newcommand{\CFFLetter}{\ensuremath{\mathcal F}}
\newcommand{\nabCFF}{\ensuremath{(\CFFParams)-\CFF}}
\newcommand{\parameterizedCFF}[2]{\ensuremath{(n,#1,#2)-\CFF}}
\newcommand{\parameterizedCFFSize}[2]{\ensuremath{\CFFSizeLetter_{n,#1,#2}}}

\newcommand{\eps}{\epsilon}
\newcommand{\set}[1]{\left\{#1\right\}}
\newcommand{\ceil}[1]{ \left\lceil{#1}\right\rceil}
\newcommand{\floor}[1]{ \left\lfloor{#1}\right\rfloor}
\newcommand{\logp}[1]{\log\parentheses{#1}}
\newcommand{\clog}[1]{ \ceil{\log{#1}}}
\newcommand{\clogp}[1]{ \ceil{\logp{#1}} }
\newcommand{\flog}[1]{ \floor{\log{#1}}}
\newcommand{\parentheses}[1]{ \left({#1}\right)}
\newcommand{\abs}[1]{ \left|{#1}\right|}

\newcommand{\cdotpa}[1]{\cdot\parentheses{#1}}
\newcommand{\inc}[1]{$#1 = #1 + 1$}
\newcommand{\dec}[1]{$#1 = #1 - 1$}
\newcommand{\range}[2][0]{#1,1,\ldots,#2}
\newcommand{\frange}[1]{\set{\range{#1}}}
\newcommand{\orange}[1]{\set{1,2\ldots,#1}}
\newcommand{\xrange}[1]{\frange{#1-1}}
\newcommand{\oneOverE}{ \frac{1}{\eps} }
\newcommand{\oneOverT}{ \frac{1}{\tau} }
\newcommand{\smallMultError}{(1+o(1))}
\newcommand{\lowerbound}{\max \set{\log W ,\frac{1}{2\epsilon+W^{-1}}}}
\newcommand{\smallEpsLowerbound}{\window\logp{\frac{1}{\weps}}}
\newcommand{\smallEpsMemoryTheta}{$\Theta\parentheses{\smallEpsMemoryConsumption}$}
\newcommand{\smallEpsMemoryConsumption}{W\cdot\logp{\frac{1}{\weps}}}

\newcommand{\largeEpsRestriction}{For any \largeEps{},}
\newcommand{\largeEps}{$\eps^{-1} \le 2W\left(1-\frac{1}{\logw}\right)$}
\newcommand{\smallEpsRestriction}{For any \smallEps{},}
\newcommand{\smallEps}{$\eps^{-1}>2W\left(1-\frac{1}{\logw}\right)=2\window(1-o(1))$}
\newcommand{\bc}{{\sc Basic-Counting}}
\newcommand{\bs}{{\sc Basic-Summing}}
\newcommand{\windowcounting}{ {\sc $(W,\epsilon)$-Window-Counting}}

\newcommand{\query}[1][] {{\sc Query}$(#1)$}
\newcommand{\add}  [1][] {{\sc Add}$(#1)$}

\newcommand{\window}{W}
\newcommand{\logw}{\log \window}
\newcommand{\flogw}{\floor{\log \window}}
\newcommand{\weps}{\window\epsilon}
\newcommand{\wt}{\window\tau}
\newcommand{\logweps}{\logp{\weps}}
\newcommand{\logwt}{\logp{\wt}}
\newcommand{\bitarray}{b}
\newcommand{\currentBlockIndex}{i}
\newcommand{\currentBlock}{\bitarray_{\currentBlockIndex}}
\newcommand{\remainder}{y}
\newcommand{\numBlocks}{k}
\newcommand{\sumOfBits}{B}
\newcommand{\blockSize}{\frac{\window}{\numBlocks}}
\newcommand{\iblockSize}{\frac{\numBlocks}{\window}}
\newcommand{\threshold}{\blockSize}
\newcommand{\halfBlock}{\frac{\window}{2\numBlocks}}
\newcommand{\blockOffset}{m}
\newcommand{\inputVariable}{x}

\newcommand{\bcTableColumnWhh}{1.5cm}
\newcommand{\bsTableColumnWidth}{1.7cm}
\newcommand{\bsExtendedTableColumnWidth}{3cm}
\newcommand{\bcExtendedTableColumnWidth}{2.8cm}
\newcommand{\bcNarrowTableColumnWidth}{1.5cm}
\newcommand{\bsNarrowTableColumnWidth}{1.5cm}
\newcommand{\bsWorstCaseTableColumnWidth}{2cm}

\newcommand{\bsrange}{ R }
\newcommand{\bsReminderPercisionParameter}{ \gamma }
\newcommand{\bsest}{ \widehat{\bssum}}
\newcommand{\bssum}{ S^W }
\newcommand{\bsFracInput}{ \inputVariable' }
\newcommand{\bserror}{ \bsrange\window\epsilon }
\newcommand{\bsfractionbits}{ \frac{\bsReminderPercisionParameter}{\epsilon} }
\newcommand{\bsReminderFractionBits}{ \upsilon}
\newcommand{\bsAnalysisTarget}{ \bssum}
\newcommand{\bsAnalysisEstimator}{ \widehat{\bsAnalysisTarget}}
\newcommand{\bsAnalysisError}{ \bsAnalysisEstimator - \bsAnalysisTarget}
\newcommand{\bsRoundingError}{ \xi}


\newcommand{\neps}{\ensuremath{\winSize\eps}}
\newcommand{\Neps}{\ensuremath{\maxWinSize\eps}}
\newcommand{\logn}{\ensuremath{\log\winSize}}
\newcommand{\logN}{\ensuremath{\log\maxWinSize}}
\newcommand{\logneps}{\ensuremath{\logp\neps}}
\newcommand{\logNeps}{\ensuremath{\logp\Neps}}
\newcommand{\oneOverEps}{\ensuremath{\frac{1}{\eps}}}
\newcommand{\winSize}{\ensuremath{n}}
\newcommand{\maxWinSize}{\ensuremath{N}}
\newcommand{\curTime}{\ensuremath{t}}
\newcommand\Tau{\mathrm{T}}
\newcommand{\offset}{\ensuremath{\mathit{offset}}}
\newcommand{\roundedOOE}{k}
\newcommand{\numLargeBlocks}{\frac{\roundedOOE}{4}}
\newcommand{\numSmallBlocks}{\frac{\roundedOOE}{2}}

\newcommand{\remove}{{\sc Remove()}}
\newcommand{\merge}[1]{{\sc Merge(#1)}}
\newcommand{\counting}{{\sc Counting}}
\newcommand{\summing}{{\sc Summing}}
\newcommand{\freq}{{\sc Frequency Estimation}}
\newcommand{\NB}{\psi}
\newcommand{\NBound}{{Z_{1 - \frac{{{\delta _s}}}{2}}}V{\varepsilon_s}^{ - 2}}
\newcommand{\matrixCellWidth}{5.8cm}
\renewcommand{\arraystretch}{1.33}
\newtheorem{observation}[remark]{Observation}

\newcommand{\ignore}[1]{}
\begin{abstract}
	
	Modern software systems are expected to be secure and contain all the latest
	features, even when new versions of software are released multiple times an
	hour. Each system may include many interacting packages. The problem of installing multiple 
	dependent packages has been extensively studied in the past, yielding some promising 
	solutions that work well in practice. However, these assume that the developers 
	declare all the dependencies and conflicts between the packages. Oftentimes, 
	the entire repository structure may not be known upfront, for example when 
	packages are developed by different vendors. In this paper we present algorithms 
	for learning dependencies, conflicts and defective packages from installation attempts. 
	Our algorithms use combinatorial data structures to generate queries that test 
	installations and discover the entire dependency structure. A query that the 
	algorithms make corresponds to trying to install a subset of packages and 
	getting a Boolean feedback on whether all constraints were satisfied in this subset. 
	Our goal is to minimize the query complexity of the algorithms. We prove 
	lower and upper bounds on the number of queries that these algorithms require 
	to make for different settings of the problem.

	\ignore{
	
	The complexity of modern software systems is continuously growing. They are
	expected to be secure and contain all the latest features, up to a 
	point where new versions of software are released multiple times an hour.
	 Each system may include many interacting packages that must be deployed and coexist. 
	 The dependencies 
between software packages may change as the system evolves. Moreover, as packages 
may be provided by different software vendors, there might be dependencies and 
conflicts between these packages that can only be discovered by installing 
the packages. Finally, there may be defective packages that would fail 
once executed in the target environment.
	
The problem of installing multiple dependent packages has been extensively
	studied in the past, yielding some promising solutions that work well in
	practice, even for thousands of packages. However, these assume that the developers declare all the dependencies and conflicts 
	between the packages. In reality, the entire repository structure may not be
	known upfront. 

	In this paper, we present algorithms for learning 
	dependencies, conflicts and defective packages from installation attempts. Our algorithms use combinatorial data structures 
	to generate queries
	that test installations and discover the entire dependency structure. 
	A query that the algorithms make corresponds to trying to install a subset
	of packages and getting a Boolean feedback on whether all constraints were
	satisfied in this subset.
	Our goal is to minimize the query complexity of the algorithms.	
	We prove lower and upper bounds on the number of
	queries that these algorithms require to make for different settings of the
	problem.
	}

\end{abstract}

\section{Introduction}
\label{sec:intro}

\ignore{
\begin{itemize}  
\item Software changes quickly and by multiple vendors. Therefore:
\item A need for packages manager, that is already implemented for OSS.
\item Most of them are SAT solvers based.
\item There are even projects like mancoosi that organize competitions between
solvers as they do not provide the best possible solution
\item All of them assume dependencies/conflict/defects are known
\item In reality, this is not the case (examples\ldots)
\item Thus, a technique for resolving all unknowns is needed
\item still NP-hard
\item Our contributions: 1 - Here we show lower bounds 
\item 2 - And algorithms that use FCCs to solve the problem for small u, r, c
\end{itemize}
}

Modern software systems are very complex modular entities, made up of many
interacting packages that must be deployed and coexist in the same context.
System administrators are reluctant to apply security patches and other updates
to packages in complex IT systems.
The reason for this hesitation is the fear of breaking the running and working
system, thus causing downtime. It is tough for such administrators to know which
upgrades to packages are ``safe'' to apply to their particular environment and
to choose a subset of upgrades to be applied. As a result, often systems are
left outdated and vulnerable for \mbox{long periods of time.}


The software upgrade problem, where we wish to determine which updates to perform,
is extensively studied~\cite{Mancinelli2006, DiCosmo2010, burrows, Trezentos,
milp}.
As many open source products such as Debian and Ubuntu operating systems are built 
from packages, some practical solutions for installing these products  have been 
developed~\cite{aptget,aptitude,smart,cupt}. 
These solutions try to find a large subset of packages that are installable together. 
Most of them either use SAT solvers or pseudo-boolean 
optimizations~\cite{Mancinelli2006, DiCosmo2010, Trezentos}. Others apply greedy 
algorithms~\cite{burrows} to derive a solution to that problem, i.e., find an 
installable subset of packages that need to be installed (or upgraded). 
These techniques assume that the dependencies and the conflicts are declared by 
the developers or can be automatically derived from package descriptors.
However, for various reasons, some information is often missing about package 
repositories. For example, when software is developed by multiple vendors, not 
all conflicts and dependencies may be known upfront. 
In addition, software components are often tested in environments different than
those in which they are eventually deployed in production, ending up with
components not working as expected.
A trivial solution to the problem of identifying such unknown relations, and to 
that of deciding on a large subset of packages to be installed is trying out 
all combinations of packages, thus discovering all the missing information. 
This solution clearly does not scale for large systems. Hence, a more effective
solution to this problem is needed.

\ignore{
 Existing solutions to the
package selection problem, which were mostly developed for open source products
such as Linux based operating systems, assume that all the dependencies and
conflicts between the different packages have been declared by the developers.

In reality, 
some packages may not be installable in the user�s environment and others may 
introduce conflicts that were not described by the developers. Finally, some 
dependencies may be missing in the description as well. To the best of our
knowledge, none of the known solutions deal with these issues. 

}

In this paper we solve the problem of detection of unknown
dependencies, conflicts, and defects while installing and upgrading a 
complex software system. Our approach addresses the dynamic
nature of dependencies between packages and the limitations that may
be prescribed by the target environment. Since some defects and constraints can
only be discovered by installing the packages, we follow a trial-and-error
strategy to learn how to install or upgrade the packages. Following this
strategy, the algorithms try to install and test different
subsets of the packages, and analyse the success/failure of installation of 
different subsets, until all dependencies, defects and conflicts are discovered. 
We choose the subsets of packages to test via a
combinatorial approach that guarantees that any combination of packages of 
predefined size will be installed and tested together while leaving out of the
installation any combination of another predefined size.
Once all the tests are finished, our technique is guaranteed to have all the 
information needed to determine if a package has a defect, or 
if there are unknown conflicts or dependencies. This allows to use much fewer tests than a trivial solution
would use, making this a feasible approach. The entire learning process is
captured by Figure~\ref{fig:process}. It starts by extracting known
dependencies structure from package descriptors and after the testing steps ends
with a complete dependencies structure.

\begin{figure}[t!]
	
	\ifdefined\fullVersion
	\includegraphics[width = 1.0\columnwidth]{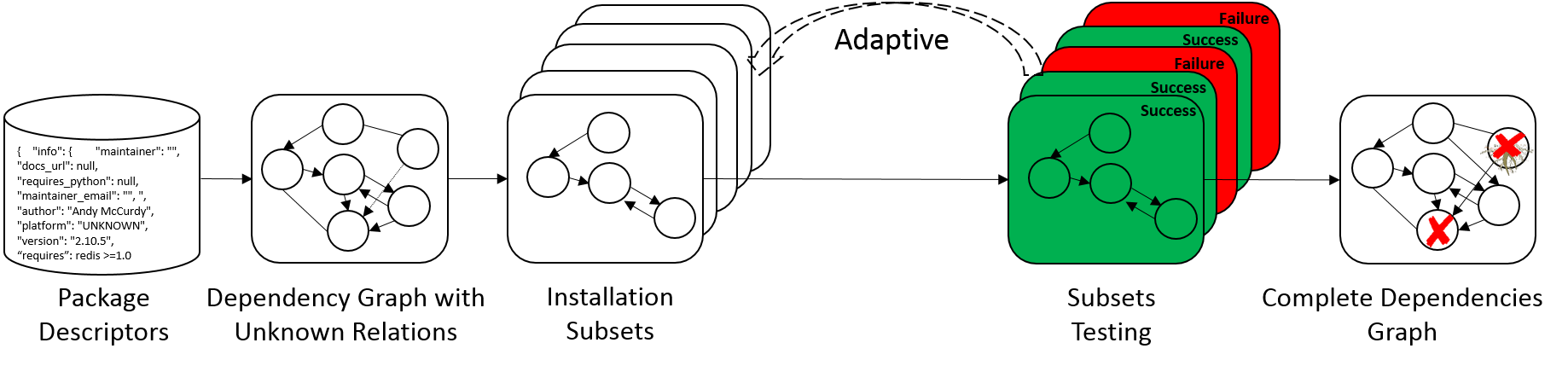}
	\caption{Learning process.}
	\else
	\centering
	\includegraphics[width = .9\columnwidth]{figs/process1v2.png}
	\vspace*{-3mm}\caption{Learning process. \vspace*{-3mm}}
	\fi
	\label{fig:process}
\end{figure}

\subsection{Contributions}
Our first contribution is the formalization of a stylized model that allows us 
to reason about the complexity of learning undocumented software constraints in a given repository. 
While previous works have considered all dependencies and conflicts to be known, 
here our goal is to handle the \mbox{undocumented package relations.}

Next, we prove lower and upper bound on the complexity of solving the
problem of resolving all the relations in the repositories graph in four
scenarios. One scenario is where the entire dependencies structure is known and
we are interested in finding the defects. The second case assumes that we have 
up to $u$ unknown dependencies and we wish to find all the
defects and the dependencies. The third case assumes that there are no
unknown dependencies, but there may be up to $c$ conflicts and $d$ defects. Finally, for the
most complex case, we assume that there can be up to $u$ unknown dependencies, up to $d$ defects,
and up to $c$ unknown conflicts 
and we find them all. For most of the scenarios we present both
\emph{non-adaptive} and \emph{adaptive} learning algorithms. Non-adaptive algorithms
work by trying out installations of subsets of packages and solve the problem at hand
based on the results of these attempts. Adaptive algorithms on the other hand
try one installation at a time and can decide which installation to try next
based on the results of the previous attempts.
\ifdefined\fullVersion
The growing complexity of the solutions for learning the relations graph for the
four scenarios is depicted in Figure~\ref{fig:contribution}, while the results are summarized in Table~\ref{tbl:summary}.

 \begin{figure}[t!]
    \centering
    \ifdefined\fullVersion
	\includegraphics[width = 0.8\columnwidth]{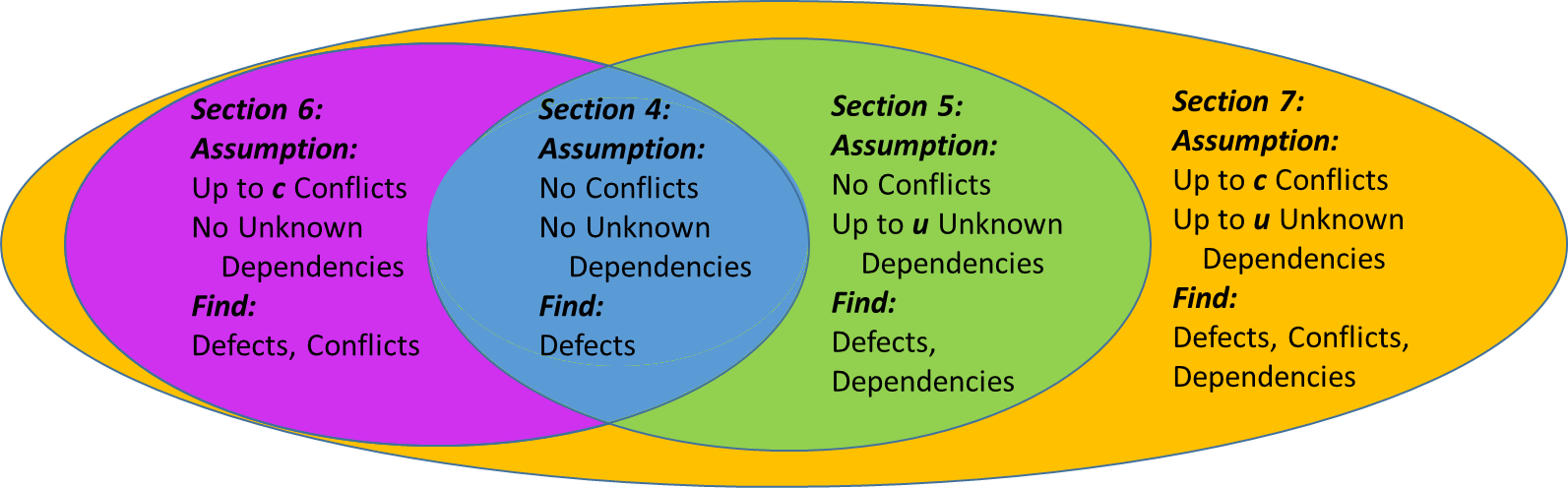}
	\caption{Growing complexity of the four scenarios considered in this work. }
	\else
	\includegraphics[width = 0.7\columnwidth]{figs/learning2.png}
	\caption{Growing complexity of the four scenarios considered in this work. \vspace*{-2mm}}
	\fi
	\label{fig:contribution}
\end{figure}

\fi

\section{Preliminaries}
For some $n\in\mathbb{N}$ we denote by $[n]\triangleq\orange{n}$ the set of integers
 smaller than or equal to $n$.
A \emph{mixed graph} $G=(V,E,A)$ consists of a set of vertices $V$, an 
undirected set of \emph{edges} $E\subseteq{V\choose 2}$, and a set of 
directed \emph{arcs} $A\subseteq V\times V$. We will use mixed graphs for modeling directed dependencies and (undirected) conflicts in software~repositories.

\ignore{Mixed graphs arise in 
several scheduling and Bayesian inference problems and will be useful for 
modeling directed dependencies and (undirected) conflicts in software~repositories.
}

\ignore{
Given a directed graph $G=(V,E)$ we denote by $SCC(G)$ the 
\emph{strongly connected components} graph of $G$ and by $C(G)$ the 
\emph{transitive closure} of $G$. \textbf{Itai: Do we use C(G) anywhere?
Maayan: we use c(v) a lot.}
}

\subsection{Learning Algorithms}
We consider algorithms that learn about properties of the underlying,
partially-unknown, graph. To that end, we evaluate the algorithms in terms of their \emph{query
complexity} -- the number of queries that they need before establishing their
answer.
That is, we assume that the algorithm has access to an \emph{oracle} that given a 
query returns a Boolean yes/no answer. In our scenario, the oracle is given a
query of whether a subset of packages can be installed, and returns yes or no
based on whether this installation is successful or not. Note that in some
settings the feedback can be more elaborate than just yes/no answer. For
example, the package management system may hint which additional packages need
to be installed. This could be used for fine tuning the subsets selection
process. However, we assume here only the minimal requirement of yes/no answers
and defer the more advanced feedback to future work.

We consider two types of algorithms -- \emph{adaptive} and \emph{non-adaptive}.
A non-adaptive algorithm is a procedure that given an input computes a set of
\emph{queries} and passes them to the oracle. When getting the Boolean
\emph{feedback} for each query it locally computes a solution to the problem. On
the other hand, adaptive algorithms are given continuous access to the oracle
and ask one query at a time. Thus, any query asked by an adaptive algorithm may
be chosen with respect to the oracle's previous answers.
\ifdefined\TACASVersion{
Non-adaptive algorithms have a parallelism advantage while adaptivity can lead
to exponentially smaller query complexity.
}
\else {
Non-adaptive algorithms have a parallelism advantage as the answers to all
queries can be computed at the same time. On the other hand, adaptivity can lead
to exponentially smaller query complexity.
Notice that we can also simulate any adaptive algorithm in a non-adaptive way
while incurring an exponential overhead, so this gap is tight.
}
\fi
\ignore{
For example, consider trying to learn an unknown number $x\in[n]$ where the
feedback on a query $y$ is one if $x\le y$. Clearly, any non-adaptive algorithm
requires $n$ queries to find $x$ while an adaptive solution can perform a binary
search using $O(\log n)$ queries.}
\subsection{Group Testing}\label{sec:GT}
In this section we provide an overview for the problem of \emph{group testing} 
that will be useful for our study for learning software relations between packages. Specifically, group testing would help us assess how many queries our algorithms need to make in order to discover all the dependencies and conflicts between software packages.

\ignore{
In group testing, we wish to identify a specific subset of $[n]$ using \emph{OR}
queries.
That is, each test corresponds to a subset $T\subseteq n$ and the feedback is True if at least one of the items in $T$ is defective.
The goal is then finding some predetermined defects set $S\subseteq [n]$ of size $|S|\le d$ (i.e., a query for $T$ will return True if $T\cap S\neq\emptyset$).
Group testing has various applications in computer science, as well as in statistics, 
biology and medicine. The study of adaptive algorithms for group testing dates back to 
1943 when Dorfman introduced the problem for identifying syphilitic soldiers~\cite{Dorfman}. 
Dorfman proposed to test equal sized soldier groups and then use individual 
tests for soldiers in the infected groups. 

This was then generalized by Li to arbitrary number of rounds, 
lowering the number of required tests to $O(d\logp{\frac{n}{d}})$~\cite{li1962sequential}. 
The constant was then improved, roughly by a factor of 1.9, using the Generalized 
Binary-Splitting (GBS) algorithm~\cite{Hwang}. 
The resulting solution is near-optimal in the sense that it requires at 
most $d-1$ more tests than the \GTLB ~information-theoretic lower bound.
Finally, a slight improvement to $0.255d+12\log_2d + 5.5$ tests over the lower 
bound was proposed when $n/d\ge 2$~\cite{allemann}.
}

 Group testing is a procedure that breaks up the task of identifying certain objects into tests on groups of items, rather than on individual ones. The study of adaptive algorithms for group testing dates back to 
1943 when Dorfman introduced the problem for identifying syphilitic soldiers~\cite{Dorfman}.  Dorfman proposed to test equal sized soldier groups and then use individual 
tests for soldiers in the infected groups. 
In our context, 
we are interested in trying to install a subset of packages. If the installation fails, then  there is either a defective package, or there is a conflict between some packages, or some dependency is missing.  
For this, we use the the Generalized 
Binary-Splitting (GBS) algorithm~\cite{Hwang} which provides the most effective solution (complexity-wise) to the group testing problem.
\subsection{Cover-Free Families}\label{sec:CFFs}
Another combinatorial structure in use in this work is that of
 $(n,a,b)$- Cover Free Families (denoted \nabCFF{}). An \nabCFF{}
 is a set of binary vectors $\CFFLetter\subseteq\set{0,1}^n$ such that on any
 $a+b$ indices we see all $a+b\choose b$ combinations of $1$s and $0$s. That
 is, we require that for any disjoint sets of indices $S_1,S_2$ of sizes $a,b$
 respectively, there exists a vector in \CFFLetter{} such that
 its $S_1$ entries are all zeros while it has ones on those of $S_2$, i.e.,
 $$\forall S_1,S_2\subseteq[n], |S_1|=a, |S_2|=b, S_1\cap
 S_2=\emptyset\implies\exists v\in\CFFLetter: v_{|S_1}=0 \wedge v_{|S_2}=1.$$

\ignore{
The problem of constructing small \CFFs{} was studied in extremal combinatorics
with applications in cryptography~\cite{boneh2004public} and graph problems such as finding an $r$-simple
$k$-path~\cite{Bshouty2017}.
We
show that \CFFs{} are intrinsically related to the problem of learning unknown
conflicts and dependencies in a software repository by showing both upper and
lower bounds that directly relate to these.}

Throughout the paper, we use
\CFFSize{} to denote the minimal size of a \nabCFF{}.
\ifdefined\fullVersion
In~\cite{CFFLB}, Stinson et al. showed that 
\ifdefined\fullVersion
$$\CFFSize{}=\begin{cases}
\Omega\parentheses{\frac{{a+b \choose b}(a+b)}{\log{a+b \choose b}}\log n} &
\mbox{if $a\le\sqrt n$} \\
\Omega\parentheses{\frac{{a+b \choose b}}{\log(a+b)}\log n} & \mbox{if $a\le n$}
\end{cases}.$$
\else
{\scriptsize $$\CFFSize{}=\begin{cases}
\Omega\parentheses{\frac{{a+b \choose b}(a+b)}{\log{a+b \choose b}}\log n} &
\mbox{if $a\le\sqrt n$} \\
\Omega\parentheses{\frac{{a+b \choose b}}{\log(a+b)}\log n} & \mbox{if $a\le n$}
\end{cases}.\vspace*{-1mm}$$}
\fi

\ignore{
Recently, a more involved analysis showed that in some cases we can improve these bounds
~\cite{NewestCFFLB,NewCFFLB}.
Since the minimal sizes of an \nabCFF{} and an \parameterizedCFF{b}{a} are clearly identical, the following expressions assume that $a\le b$; however, in the rest of the paper this is not necessarily the case.
}
\fi
We can efficiently construct a \CFF{} probabilistically by creating a set of
${a+b\choose a}(a+b)^{O(1)}\log n$ binary vectors of length $n$, where each
bit is set to 1 independently with probability $\frac{b}{a+b}$. The resulting
randomized set is an \nabCFF{} with high probability.
The best known deterministic construction for \CFFs{}~\cite{Bshouty2017}, which is also computed in linear time,  provides an upper bound of:
\ifdefined\fullVersion
$$\CFFSize{}\le\begin{cases}
{b^{a+1+o(1)}\log n} & \mbox{if $a = O(1)$} \\
{\parentheses{\frac{a+b}{a}}^{a+o(a)}\log n} & \mbox{if $a=\omega(1)\wedge a=o(b)$}\\
{2^{(a+b)H(\frac{a}{a+b})+o(b)}\log n} & \mbox{if $a=\Theta(b)$}\\
\end{cases},$$
\else
{\scriptsize 
$$\CFFSize{}\le\begin{cases}
{b^{a+1+o(1)}\log n} & \mbox{if $a = O(1)$} \\
{\parentheses{\frac{a+b}{a}}^{a+o(a)}\log n} & \mbox{if $a=\omega(1)\wedge a=o(b)$}\\
{2^{(a+b)H(\frac{a}{a+b})+o(b)}\log n} & \mbox{if $a=\Theta(b)$}\\
\end{cases},\vspace*{-3mm}$$
}
\fi
where $H(x)\triangleq-x\log_2 x - (1-x)\log_2 (1-x)$ is the binary entropy function. 
In order to avoid using these cumbersome expressions, we will hereafter express our upper and lower bounds as a function of \CFFSize{} for different values of $a,b$.

\section{Model}
\label{sec:randomized}
This section formally defines the problems we are
interested in solving, using graph theory. It starts by presenting the basic
terminology that we use to describe relations between packages in software
repositories. It then presents two learning objectives that we are considering. It also gives a summary of the notations used throughout the paper.


\subsection{Basic Terminology}
\label{sec:terminology}
We consider a set of \emph{packages} $P$ that represents the modules in our
repository.
An \emph{installation} is a set of packages $I\subseteq P$; intuitively, an installation can be successful or not depending on whether all dependency, conflict and defect constraints are satisfied as we formally define below.
A \emph{dependency} $(q,p)\in P^2$ is an ordered pair which means that any
installation that includes $q$ but excludes $p$ will fail. Similarly, a \emph{conflict} $\{p,q\}\in{P\choose 2}$ implies that any installation with both $p$ and $q$ will fail. We assume throughout this work that there are no alternatives to the declared dependencies, i.e. if $q$ depends on $p$ it can not be installed without $p$ by using another implementation $p_1$.
Similar definitions were introduced in previous
works~\cite{DiCosmo2010,Mancinelli2006}. The main difference is that prior
solutions assumed that all dependencies and conflicts are known while we address
the problem of \emph{learning} these using an oracle. That is, we assume that
one can try any installation and get a \emph{feedback} on whether it succeeded.
Using this feedback, our goal is to learn the unknown dependencies and conflicts
while minimizing the number of installation attempts. We also consider the concept 
of \emph{defects} -- packages that can not be a part of any successful installation. 
This can be due to a broken release, inconsistencies, etc. Notice that this
means that if a package $p$ depends on a defective module $q$, then $p$ could never be successfully 
installed and thus is also a defect. We also consider the notion of \emph{root defects} 
which are the root cause for an install to fail. In the example above, where $p$ depends 
on a defect $q$, we call $q$ a root defect. Formally, a root defect is a defective 
package for which all of the modules it depends on are not defects.

We model relations within a repository using a mixed graph $G=(P,C, K
\cup U)$ where $K$ is the set of known (directed) dependencies, $U$ is the set of unknown (directed) dependencies and
$C$ is the set of unknown (undirected) conflicts. Defects are modeled as a set
$D\subseteq P$ of packages that can not be installed or fail to work once
installed.
Notice that our definition of defects implies that $D$ has no incoming arcs. That is, $(P \times D)  \cap
(K\cup U)=\emptyset$.

Consider a cycle of known dependencies $p_1\to p_2\to \ldots\to p_z\to p_1$.
This implies that any successful installation must either install \emph{all} of
$p_1,p_2,\ldots,p_z$ or none of them. This allows us to contract these into a
single ``super-package'' whose installation is equivalent to that of all of
them. That is, we can consider the strongly connected components graph instead
of that of the original repository.\footnote{The exception here is that if one
of the packages in the component is a root defect, we will only identify that
one of the packages in the strongly connected component is defective. We emphasize that even without
contracting strongly connected components, 
these are indistinguishable and thus the root defect 
\mbox{cannot be learned in this model.}}
Thus, we henceforth 
assume that the induced digraph $G_K\triangleq(P, K)$ that contains only the known dependencies is \emph{acyclic}.

\ignore{Given a graph $G=(V,E)$ and a vertex $v\in V$, we denote by $c(v)$ the
\emph{transitive closure} of $v$. That is, $c(v)= \set{u | u\mbox{ is reachable
from $v$ in } G}$. In our framework, one cannot distinguish between packages in
the same connected component that has an unknown dependency. That is, assume that $p_1,p_2$ are in the same connected 
component in $c(P,K\cup U)$; using binary feedback one can never conclude 
whether $(p_1,p)\in U$ or $(p_2,p)\in U$ for some package $p\in P$. Thus, it
only makes sense to try to learn the transitive closure of the dependency graph. Further, when trying to install a package or the largest set of updates, the closure graph of the dependencies is the desired output, as we only wish to know which packages depend on which. Similarly, we cannot hope to distinguish the  $\set{p,p_1}\in C$ case from $\set{p,p_2}\in C$. Again, in practice all we need to know is that $p_1,p_2$ must be installed together and that an installation cannot contain both $p$ and $p_1,p_2$.
This motivates us to set as a goal to learn the strongly connected component graph of the dependency closure, and find the conflicts between components.
}

In our framework, one cannot distinguish between packages in
the same connected component that has an unknown dependency. 
That is, assume that $p_1,p_2$ are in the same connected 
component; using binary feedback one can never conclude 
whether $(p,p_1)\in U$ or $(p,p_2)\in U$ for some package $p\in P$. 
Thus, it
only makes sense to try to learn the transitive closure of the dependency graph. 
Further, when trying to install a package or the largest set of updates, 
the closure graph of the dependencies is the desired output, as we 
only wish to know which packages depend on which. We denote by $C(G)$ the
\emph{transitive closure} of a given graph $G=(V,E)$. That is, the vertex set of
$C(G)$ is $V$ and the edge set is $\set{(v_1,v_2)\mid v_2\mbox{ is reachable from $v_1$ in } G}$.

Similarly, we 
cannot hope to distinguish the  $\set{p,p_1}\in C$ case from $\set{p,p_2}\in C$. 
Again, in practice all we need to know is that $p_1,p_2$ must be installed 
together and that an installation cannot contain both $p$ and $p_1,p_2$.
This motivates us to set as a goal to learn the strongly connected 
component graph of the dependency closure, and find the conflicts between components.
\ifdefined\fullVersion

\fi
Table~\ref{tbl:notations} summarizes the notations used in this work.
\begin{table}[h]
	\centering
	\ifdefined\TACASVersion
	\scriptsize
	\else
    \small
    \fi
	
	\begin{tabular}{|c|l|}
		
		\hline
		Symbol & Meaning \tabularnewline
		\hline\hline
		\ifdefined\fullVersion
		$G=(V,E,A)$ & Mixed graph, with undirected edges E and directed edges A
		\tabularnewline
		\hline
		\fi
		 $G=(P,C, K \cup U)$ & Mixed graph, with packages $P$ as nodes, conflicts $C$,
		 \tabularnewline &  known dependencies $K$, and unknown dependencies $U$
		\tabularnewline
		\hline
		 $(n,a,b)-\text{CFF}$  $\CFFLetter$ & Cover-Free Family, where each vector
		 has $0$ at $a$ indexes, and $1$ at $b$ indexes
		  \tabularnewline
		\hline
		$\CFFSize{}$ & Size of best known deterministically constructed $(n,a,b)$-CFF 
		\tabularnewline
		\hline
		$P$ & Installable packages or modules in a software repository 
		\tabularnewline
		\hline
		$I$ & An installation of packages, subset of $P$ 
		\tabularnewline
		\hline
		$D$ & Uninstallable packages (defects) 
		\tabularnewline
		\hline
		$C(G)$ &  Transitive closure of a graph in a graph 
		\tabularnewline
		
		\hline
		$G_K=(P, K)$ & Acyclic graph with only known dependencies
		\tabularnewline
		\hline
		$T(p)$ & Set of tests that tried to install package $p$  \tabularnewline
		\hline
		$S(p)$ & Set of successful installations that included $p$
		\tabularnewline
		\hline
		$r$ / $u$ / $c$ & Bounds on the number of root defects/unknown dependencies/conflicts  \tabularnewline
		\hline

	\end{tabular}
	\caption{List of Symbols}
	\label{tbl:notations}
	\ifdefined\TACASVersion
	\vspace*{-6mm}
	\fi
\end{table}

\subsection{Learning Objectives and Problem Definitions}
\ignore{
\textbf{TODO: see if we're sticking to only 2 problems}
In this paper we consider three learning problem variants:
\begin{enumerate}
	\item \maxProblem{}: Given $G_K$, the induced known dependency digraph, 
	and bounds $c, u, d$ such that the number of conflicts is at most $c$, the
	number of unknown dependencies is at most $u$, and the number of defective
	packages is at most $d$, find a \emph{maximum}-size set of packages
	$P_{max}\subseteq P$ that can be successfully installed.
	\item \label{obj:packageInstallation} \minProblem{}: Given $G_K$, a target package
	 $p\in P$ and bounds $c, u, d$ as above, find a \emph{minimum}-size successful
	 installation $I_p$ that contains $p$. \textbf{Itai: remove this ? }
	\item \learnAll{}: Given $G_K$ and bounds $c, u, d$ as above, return the 
	mixed graph $\bar G=(V, E, A)$ such that $V$ contains strongly connected
	components of all the packages, with the defective packages marked as such. 
	$E$ are all the known and unknown conflicts between the nodes in $V$, and $A$
	are all known and unknown dependencies.
\end{enumerate}
The first objective is motivated by Continous Deployment (CD) of
software~\cite{cd}.
In CD, we have multiple teams working separately on modules that be depend on
each other, conflict, or just cease to work. Thus, at each point in time we wish to 
enable as many modules as possible to bring the best user experience.
 \textbf{TODO: we need a real motivation here..}.

\textbf{Itai: remove, assuming we remove
the second problem.} Objective~\ref{obj:packageInstallation} addresses the
challenge a user faces when trying to install a software package. Automatic tools such as PIP,?? and ??, simplify the installation process by recursively installing all the modules that the target package is known to depend on. However, this does not always succeed due to hidden conflicts, dependencies, or defects.

The last objective allows the system administrators to learn the exact state
of the repository. As our main metric is the query complexity, a solution to this problem 
implies that we can also solve the first two problems by local computation. Hence, 
this problem is the hardest and any lower bound on the other problems is 
directly applicable to it as well.
}

In this paper our objective is to solve two learning problem variants:
\begin{enumerate}[leftmargin=*]
	\item \maxProblem{}: Given $G_K$, the induced known dependency digraph, 
	and bounds $c, u, d$, 
	\ifdefined\fullVersion such that the number of conflicts is at most $c$,  the
	number of unknown dependencies is at most $u$, and the number of defective
	packages is at most $d$,\fi  find a \emph{maximum}-size set of packages
	$P_{max}\subseteq P$ that can be successfully installed.
	\item \learnAll{}: Given $G_K$ and bounds $c, u, d$ as above, return the 
	mixed graph $\bar G=(V, E, A)$ such that $V$ contains strongly connected
	components of all the packages, with the defective packages marked as such. 
	$E$ are all the known and unknown conflicts between the nodes in $V$, and $A$
	are all known and unknown dependencies. A sample input and output of 
\learnAll{} is shown in Figure~\ref{fig:example}. Note that by solving
\learnAll{} one gets an answer to \maxProblem{}~as~well.
\end{enumerate}

The first objective is motivated by security updates, where one receives
updates from multiple sources, that may depend on each other, conflict,
or misbehave in the target system. 
\ifdefined\fullVersion
Thus, we wish to find the largest possible
subset of patches that can be safely installed, in order to make \mbox{the system as
secure as possible.}
\else
Thus, we wish to find the largest possible
subset of patches that can be safely installed, to make \mbox{the system as
secure as possible.}
\fi
\ignore{The first objective is motivated by Continuous Deployment (CD) of
software~\cite{cd}.
In CD, we have multiple teams working separately on modules that may depend on
each other, conflict, or just cease to work. Thus, at each point in time we wish to 
enable as many modules as possible to bring the best user experience.}

The second objective allows the system administrators to learn the exact state
of the repository. As our main metric is the query complexity, 
a solution to this problem 
implies that we can also solve the first problem by local computation. Hence, 
\ifdefined\fullVersion this problem is the hardest and \fi any lower bound on the first problem is 
directly applicable to it as well.

\begin{figure*}[]
	\begin{tabular}{c|c|c}
		\subfloat[Legend]{\includegraphics[width = 0.28\columnwidth]
			{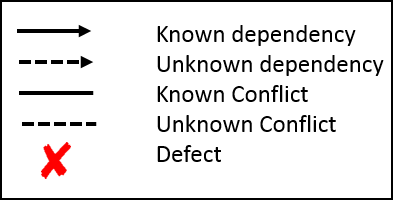}} &	
		\subfloat[Repository]{\includegraphics[width = 0.32\columnwidth]
		{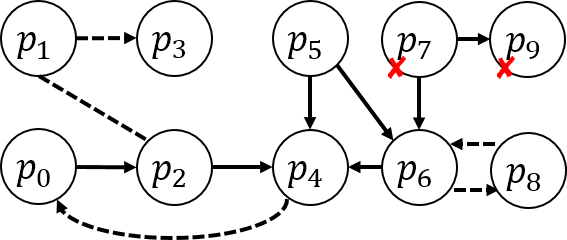}\label{subfig:input}} &			
		\subfloat[Output]{\includegraphics[width = 0.32\columnwidth]
			{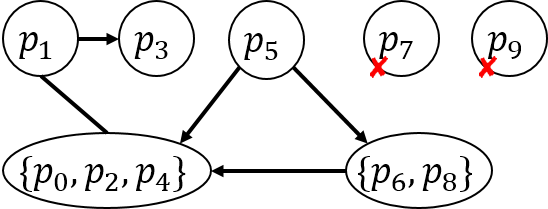}\label{subfig:output}
		} 
	\end{tabular}
	\caption{An example of \learnAll{}.  
	The input is depicted in~\ref{subfig:input}. The output is shown
	in~\ref{subfig:output}, and contains the strongly
	connected components of the actual dependency graph along 
	with the full specification of the dependencies, defects and conflicts.
	}
	\label{fig:example}
\end{figure*}
\section{Learning when all dependencies are known}
In this section, we assume that all the dependencies are known, no conflicts
exist, but the repository may contain some root defects that could fail an
installation. Specifically, we allow at most $r$ root defects, while there can be \mbox{as many as $n$ defects overall.}

\begin{theorem}\label{thm:GTLB}
\ifdefined\TACASVersion
\vspace*{-1.5mm}
\fi
	Denote $B\triangleq \ceil{\log_2\parentheses{\sum_{i=0}^{r}{n\choose
	i}}}$; any algorithm that solves group testing on $n$ items 
	and at most $r$ defects requires $B$ queries.
\ifdefined\TACASVersion
\vspace*{-1.5mm}
\fi
\end{theorem}

\begin{proof}
We start by observing that if there are no known dependencies
($K=\emptyset$), the problem reduces to group testing over $n$ items and at most
$r$ defects. Thus, the lower bound proven in~\cite{li1962sequential} applies to our problem as well.
\vspace*{-1mm}
\end{proof}

We proceed with an algorithm for the case where $K \neq \emptyset$, which is
based on the Generalized Binary Splitting (GBS) method~\cite{Hwang} mentioned
above.
Specifically, we show that its routine can be implemented despite the constraints imposed by the dependencies.
We show that in this case, the number of tests required to learn the root defects and solve \learnAll{adaptively} is similar to that of group testing. 
Intuitively, the GBS algorithm arbitrarily chooses the sets to test while
determining only their size. Here, we use the set of known dependencies $K$ to
determine which packages to try at each point.
In order to find a defect in a set of $2^\alpha$ packages (for an $\alpha\in
\mathbb{N}$), we first compute a topological sort $L$ on its
vertices~\cite{Kahn}, whereas the vertices with no outgoing dependencies have the highest indexes.
That is possible as $K$ contains no cycles (as explained in Section~\ref{sec:terminology}). Then we
first test the $2^{\alpha-1}$ packages with the lowest indices in $L$. 
If the test fails, we recurse on the tested vertices. Otherwise, we recurse on
the remaining packages \emph{while adding the non-defective vertices} to all future installations. 
That is, since we know that these $2^{\alpha-1}$ packages are non-defective, we can 
safely add them to all other queries, thereby resolving dependencies of the
other packages. Next, we follow a similar procedure for the main GBS iteration.
Our algorithm starts by selecting the $2^\alpha$ packages with the highest index
in $L$ and thus ensures that they do not depend on other modules. If the test fails, we can use the above to find a defect using $\alpha$ queries.
Once a root defect has been identified, we remove all packages that depend on it
as they are considered as defects.
On the other hand, if the test succeeds, we repeat while adding these packages
to future installations. Finally, if $n\le 2r-2$ we can individually test 
each package according to their index in $L$. Throughout the algorithm, we 
maintain the reservoir that for any two packages $p,q\in P$ such that
$L[p]>L[q]$, $p$ is tested without $q$ only if $q$ is identified as a defect.
Thus, we never test a package without installing all modules it depends on.
\ifdefined\TACASVersion
We provide a detailed pseudo code of our method in the full version of the paper~\cite{fullVersion}.
\else
We provide a pseudo code of our method in Algorithm~\ref{alg:basic}.


\begin{algorithm}[t]
    \ifdefined\fullVersion
	   \footnotesize
	\else
	        \scriptsize
	\fi

	\begin{algorithmic}[1]
		\Function{FindDefects}{$P,K,r$}\Comment{Find at most $r$ root defects in
		$P$} \State $S\gets \emptyset$ \Comment{The set of identified non-defect packages}
		\State $R \gets \emptyset$ \Comment{The set of identified root defects}
		\State $\mathfrak r \gets r$ \Comment{A bound on the number of unidentified root defects} \State $X \gets P$ \Comment{The rest of untested packages}
		\State $L \gets \text{Topolocial Sort}(P,K)$
		\While {$|X| > 2\mathfrak r-2$}
		\State $l \gets |X|-\mathfrak r+1$
		\State $\alpha \gets \floor{\log_2 l/\mathfrak r}$ \Comment{As defined in GBS
		procedure~\cite{Hwang}} 
		\State $T \gets \mbox{A $2^\alpha$-sized subset of $X$ with minimal
		$L$ indices}$
		\If{$(T\cup S)$ fails} \Comment{Test $T\cup S$}
		\State $p \gets \mbox{\sc {FindSingleDefect}}(T)$ \Comment{Find a root defect using $\alpha$ tests}
		\State $R\gets R\cup \set{p}$
		\State $X\gets  X \setminus\set{p'\in P\mid p'\mbox{ is reachable from $p$ in
		}(P,K)}$ \Comment{Remove all packages that
		depend on $p$} \State $\mathfrak r\gets \mathfrak r-1$ \Else\Comment{If the test succeeded}
		\State $S \gets S \cup T$
		\EndIf
		\State $X\gets  X \setminus S$ \Comment{Remove from $X$ packages
		discovered as working}
		\EndWhile
		\For {$p\in X$, in an increasing order of $L$}
		\If {$(\set{p}\cup S)$ fails} \Comment{Test $\set{p}\cup S$}
		\State $R\gets R\cup \set{p}$
		\State $X\gets  X \setminus\set{p'\in P\mid p'\mbox{ is reachable from $p$ in
		}(P,K)}$\Comment{Can be computed using BFS}
		\Else
		\State $S\gets S \cup \set{p}$
		\EndIf
		\EndFor
		\State\Return $R$
		\EndFunction
		
		\Function{FindSingleDefect}{$T$}
		\State $A \gets T$ \Comment{The set of suspicious packages}
		\While {$|A|>1$}
		\State $B \gets \mbox{An $|A|/2$-sized subset of $A$ with minimal $L$ indices}$            
		\If{$(B\cup S)$ fails}\Comment{Test $B\cup S$}
		\State $A\gets B$
		\Else\Comment{If the test succeeded}
		\State $S \gets S \cup B$
		\State $A \gets A\setminus B$                
		\EndIf    
		\EndWhile
		\State\Return $p\in A$ \Comment{The remaining package is a root defect.}
		\EndFunction        
	\end{algorithmic}
	\normalsize
	\caption{Identifying root defects given all dependencies}
	\label{alg:basic}
\end{algorithm}
\fi
Since in each test all dependencies are satisfied, and as we follow GBS at each iteration, we conclude the correctness and query complexity of our algorithm. 
\ifdefined\TACASVersion
\begin{theorem}\label{thm:knownDependenciesUpperBound}\vspace*{-1mm}
	Our algorithm finds the root defects (and thus solves \learnAll{}) using at most $B+r-1$ queries, where $B$ is the lower bound from  Theorem~\ref{thm:GTLB}.\vspace*{-1mm}
\end{theorem}
\else
This also proves the following theorem:
\begin{theorem}\label{thm:knownDependenciesUpperBound}
	Algorithm~\ref{alg:basic} finds the root defects (and thus solves \learnAll{}) using at most $B+r-1$ queries, where $B$ is the lower bound from  Theorem~\ref{thm:GTLB}.
\end{theorem}
\fi
\section{Learning with Unknown~Dependencies}\label{sec:unknownDeps}
In the previous section, we assumed that all the dependencies are known and identified 
the defective packages. Here, we assume that some of the dependencies in the repository 
may not be documented. Thus, the GBS variant we proposed no longer works, and we need a different solution.

We now show that even if there exists no more than a single unknown dependency
and a single root defect, no algorithm with a sub-linear number of queries exists even when adaptivity is allowed.
Note that in this section we solve only \maxProblem{}. The more difficult
problem, \learnAll{}, needs to be solved using algorithms presented in
Section~\ref{sec:all}.

    \begin{theorem}\label{thm:oneRootOneUnknownDep}
	Any \textbf{adaptive} algorithm that solves \maxProblem{} must make at least $n$
	queries in presence of unknown dependencies \mbox{and root defects.}
	\end{theorem}

\begin{proof}
	Denote $P\triangleq\set{a_1,\ldots,a_n}$ and consider the directed path graph 
	given by 
	$K\triangleq\{(a_i,a_{i-1})\mid $$ i\in\set{2,\ldots,n}\}$ (as illustrated in
	Figure~\ref{fig:rootLB}).
	Any installation 
	considered by the algorithm is either a prefix of the line, i.e., 
	$\set{a_1,a_2,\ldots,a_i}$ for some $i\in\set{1,2,\ldots n}$, or an
	installation that does not take all of the prerequisites into consideration
	(and thus fails).
	In the former case, let us assume, by contradiction, that there exists an
	$i\in\set{1,2,\ldots n}$ such that the $\set{a_1,a_2,\ldots,a_i}$ installation was not tested by the algorithm. We
	use an adversary argument and show that the algorithm cannot distinguish
	between two problem instances with a distinct solution.
	For this, consider 
	$U=\set{(a_1,a_i)}$ and assume that $a_{i+1}$ is a root
	defect, as illustrated in Figure~\ref{fig:rootLB} 
	. The only installation that could work
	is $\set{a_1,a_2,\ldots,a_i}$, which the algorithm did not test. Thus, all
	queries made by the algorithm came back negative. There is no way for 
	the algorithm to know whether the solution should be $\set{a_1,a_2,\ldots,a_i}$
	or $\emptyset$ which reflects the case where $a_1$ is a root~defect.
\end{proof}

\begin{figure}[]
	\center
	\includegraphics[width = 0.8\columnwidth]{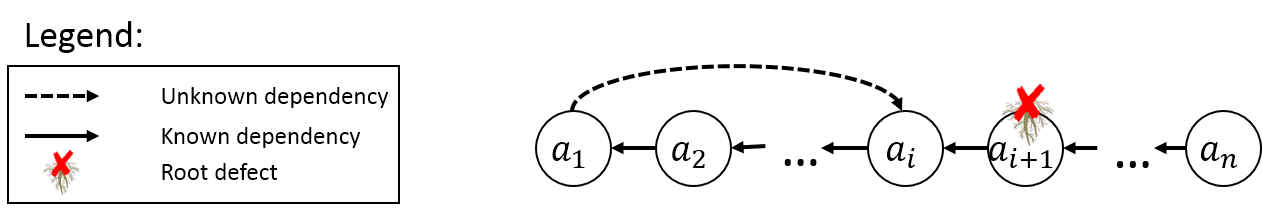}
	
	\caption{If the algorithm does not test the installation $\set{1,\ldots i}$, for some $i\in\set{1,\ldots n}$, then it cannot distinguish between the case where $a_1$ is a root defect and thus no installation succeeds and the case where $a_1$ has unknown dependency on $a_i$ and $a_i+1$ is a root defect.
	}
	\label{fig:rootLB}
\end{figure}

As shown in the theorem above, if we do not bound the number of defects, no algorithm can efficiently solve the problem even for a single root defect. Recall that a defect is a package that cannot be included in a successful installation. This can be either due to a bug in the package itself or due to a dependency on a corrupted package.
Thus, we hereafter consider a bound $d$ on the number of defects. 
Intuitively, we will show that if $d$ is small, the problem becomes 
tractable again. Note that in Theorem~\ref{thm:GTLB} no bound was imposed on the
number of defects, but rather on the quantity of \textit{root~defects}.

We proceed with a lower bound on the number of queries required by any non-adaptive algorithm when the number of defects is bound by $d$. Recall that
$\CFFSize{}$ is the size of a \nabCFF{} as described in Section~\ref{sec:CFFs}.

\ifdefined\TACASVersion
\begin{theorem}\label{thm:unknownDependenciesLB}(*)
Assume that the repository contains at most $d$ defects and $u$ unknown
dependencies.
Any \textbf{non-adaptive} algorithm for \maxProblem{} must make at least $\parameterizedCFFSize{d-1}{u+1}$ queries.\footnote{Proofs of the results marked (*) are omitted due to space limitation and appear in the full version~\cite{fullVersion}.}
\end{theorem}
\else
\begin{theorem}\label{thm:unknownDependenciesLB}
Assume that the repository contains at most $d$ defects and $u$ unknown
dependencies.
Any \textbf{non-adaptive} algorithm for \maxProblem{} must make at least $\parameterizedCFFSize{d-1}{u+1}$ queries.
\end{theorem}
\begin{proof}
Consider $K=\emptyset$, i.e., a repository with no known dependencies. Assume that an algorithm 
tries less than $\parameterizedCFFSize{d-1}{u+1}$ installations before its output. 
Then there exists a pair of disjoint package-sets $S_1,S_2$, of sizes $d-1$ and $u+1$, such that no 
attempted installation includes all of $S_2$ and none of the packages in $S_1$. 
We show that in 
this case, it cannot possibly find the maximal installation in the worst case.
We denote $S_1\triangleq\set{x_1,\ldots x_{d-1}}$ and 
$S_2\triangleq\set{a_1,\ldots a_u}\cup \set{p}$. 
The set of (unknown) dependencies is $U\triangleq\set{(p,a_i)\mid
i\in\set{1,\ldots,u}}$ and the set of defects includes $\set{x_1,\ldots x_{d-1}}$. An illustration of the setting appears in Figure~\ref{fig:unknownDependenciesLB}.

Note that thus far $|S_1|=d-1,|S_2|=u+1, |U|= u, |D|=d-1$. We now claim that the algorithm cannot 
possibly distinguish between the case where $p$ is a defect and the case where it is not. 
Notice that in order for $p$ to be a part of a successful installation, the
installation must contain $S_2$ and none of the packages in $S_1$. Thus, all
installations attempted by the algorithm were either unsuccessful or did not contain $p$. Since the same test results would be obtained regardless of whether $p$ is defective, we conclude that the algorithm cannot determine the maximal installation as it must contain $p$ if it is not defective.
\end{proof}

\begin{figure}[]
	\center
	\includegraphics[width =\columnwidth]{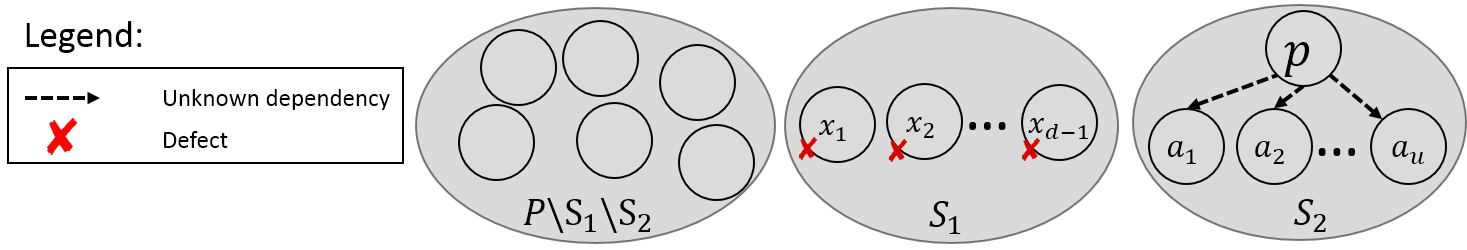}

	\caption{If the algorithm does not test an installation that includes
	 ${p,a_1,\ldots,a_u}$ and excludes ${x_1,\ldots,x_{d-1}}$ then it cannot
	  determine whether $p$ is defective and thus cannot solve \maxProblem{}.
	}
	\label{fig:unknownDependenciesLB}
\end{figure}

\fi
We now provide a non-adaptive algorithm for \maxProblem{} that requires 
$\parameterizedCFFSize{d}{u+1}$ queries. Notice that it is optimal up to the (-1) 
factor in the first parameter. We note that an algorithm for \learnAll{} is presented in the following section, but here we provide a more efficient algorithm for the simpler \maxProblem{} problem when no conflicts exist.

Intuitively, we construct a $\parameterizedCFF{d}{u+1}$, factor in the known dependencies, and get a  set of tests that will later allow us to infer the maximal installation. 
The improvement in query complexity over the \learnAll{} algorithm presented below is that when no conflicts exist, finding the maximal installation is equivalent to identifying defects.
\begin{observation}
When no conflicts exist, the set of non-defective packages is the maximal installation.
\end{observation}

Henceforth, we interchangeably refer to $n$-sized binary vectors as subsets of $P$. Fixing a canonical enumeration of $P$, we say that the $i^{\mathit{th}}$ package is in a vector $v\in\set{0,1}^n$ if its $i^{\mathit{th}}$ bit is set..
\ifdefined\fullVersion
\vbox{
\fi
Formally, we construct a $\parameterizedCFF{d}{u + 1}$ $\CFFLetter$ and define the test set as 
$\mathcal T \triangleq \set{c(v)\mid v\in \CFFLetter}$, 
where $c:\set{0,1}^n\to\set{0,1}^n$ \mbox{is defined as}
$$c(v)\triangleq v
\cup \set{p\in P\mid \exists p'\in v:
\mbox{$p$ is reachable from $p'$ via known edges}}.$$
\ifdefined\fullVersion
}
\fi
That is, given a vector $v\in\CFFLetter$ we create a test $c(v)$ that contains all 
packages whose bit is set, together with those that are prerequisites to some set-bit package. 
For example, if $P=\set{1,2,3,4}, K=\set{(2,1),(3,2)}$ and
$v=\langle0,1,0,0\rangle$ then we test the installation $c(v)=\set{1,2}$. Notice that we can compute the transitive 
closure of the dependencies graph once in $O(n^2)$ time and use it to compute
$c(v)$ in linear time for any vector $v\in\set{0,1}^n$. 

After testing all installations, for all tests $t\in \mathcal T$ we receive a
feedback $f(t)$ of whether $t$ succeeded. We now prove that given $\set{f(t)}_{t\in\mathcal T}$ we can identify all defects.
First, for each package $p$ we define its set of successful installations.
\begin{definition}
	Given $p\in P$ let $T(p)\triangleq \set{t\in\mathcal T\mid p\in t}$ be the set of tests that installed $p$ and $S(p)=\set{t\in T(p)\mid f(t)=1}$ be the \mbox{installations that were successful.}
\end{definition}
We first show that a package is defective only if it was not successfully installed in any of the~tests.
\ifdefined\TACASVersion
\begin{lemma}\label{lem:defectiveness}(*)
	A package $p\in P$ is defective if and only if $S(p)=\emptyset$.
\end{lemma}
\else
\begin{lemma}\label{lem:defectiveness}
	A package $p\in P$ is defective if and only if $S(p)=\emptyset$.
\end{lemma}
\begin{proof}
	Recall that by definition a package is defective if there exists no successful
	installation that contains it. Hence, $p\in D$ immediately implies that $S(p)=\emptyset$.
	Our goal here is to show the converse -- that $S(p)=\emptyset$ implies that $p$
	cannot be a part of any successful installation, including the ones that were
	not tested by the algorithm. Thus, we assume that $p\notin D$, and show a test
	that was necessarily included in $\mathcal T$ and succeeds.
	
	We now construct disjoint sets $S_1,S_2\subseteq P$ of sizes $d$ and $u + 1$. 
	Here, we choose $S_1\triangleq D$ to be the set of defects.
	Next, we define $S_2\triangleq\set{p}\cup\set{p'\in (P\setminus D)\mid \exists
	p''\in P, (p'',p')\in U}$.
	In other words, we add to $S_2$ the package $p$ and every \emph{non-defective}
	package $p'$ that is a prerequisite to a package $p''$ and the dependency
	$(p'',p')$ was missing from $K$.
	By adding to $S_2$ only non-defective packages we guarantee that $S_1\cap S_2=\emptyset$ as required.
	Also, we added at most $d$ packages to $S_1$ and at most $u+1$ packages to $S_2$.\footnote{We can add arbitrary packages to $S_1$ and $S_2$ to make their sizes exactly $d$ and $u+1$ if needed.}
	
	While we cannot determine $S_1$ and $S_2$ in advance,
	 from $\CFFLetter$'s properties, we are guaranteed that there exists a vector $v\in\CFFLetter$ such that $S_1\cap v=\emptyset$ and $S_2\subseteq v$. 
	Now, observe that if $p$ is not defective then $c(v)$ must pass as it contains no defects and satisfies every dependency. The known dependencies are satisfied due to the propagation of $K$ in $c(v)$, and the unknown dependencies are satisfied as they are included in $S_2$.
	Thus, we established that if $p$ is not defective then $c(v)$ must pass as all defects were excluded and all dependencies satisfied.
\end{proof}
\fi
The pseudo code of our method for \maxProblem{} solution when no conflicts
are present is given in Algorithm~\ref{alg:unknowns_no_conflicts}.
\begin{algorithm}[t]
    \ifdefined\fullVersion
	    \footnotesize
	\else
	\scriptsize
	\fi
	\begin{algorithmic}[1]
		\Function{FindMaxSubRepository}{$P,K,u,d$}
		\ifdefined\TACASVersion
    		\Comment{\hspace*{-1mm}At most $u$ unknown dependencies and $d$~defects}
    	\else
    		\Comment{At most $u$ unknown dependencies and $d$~defects}
    	\fi
		\State $D \gets \emptyset$ \Comment{The set of identified defects}
		\State $\CFFLetter \gets CFF(n,d,u+1)$ \Comment{CFF construction}
		\State $\mathcal T \gets \set{c(v)\mid v\in \CFFLetter}$
		\Comment{Test vectors generated based on all known dependencies}
		\For {$p \in P$ }
			\State $S(p) \gets \emptyset$ \Comment{Successful installations that
			included $p$}
		\EndFor
		\For {$t \in \mathcal T$}
		\State $f(t) \gets Feedback(t)$ \Comment{Get feedback from oracle}
		\If {$f(t)$ = 1} \Comment{Successful installation}
		\For {$p$ tested as part of $t$}
			\State $S(p) \gets S(p) \cup t$
		\EndFor        
		\EndIf
		\EndFor
		
		\For {$p \in P$ }
			\If {$S(p) = \emptyset$} \Comment{No successful installation exists for $p$}
			\State $D \gets D \cup p$
			\EndIf
		\EndFor
		
		\State\Return $P \setminus D$
		\EndFunction
		
	\end{algorithmic}
	\normalsize
	\caption{\maxProblem{} with defects, unkwnonwn dependencies, and no conflicts}
	\label{alg:unknowns_no_conflicts}
\end{algorithm}
We conclude an upper bound on the non-adaptive query complexity of \maxProblem{}. 
\begin{theorem}\label{thm:unknownDependenciesUB}
Assume that the repository contains at most $d$ defects and $u$ unknown
dependencies.
There exists a \textbf{non-adaptive} algorithm that solves \maxProblem{} using $\parameterizedCFFSize{d}{u+1}$ queries.
\end{theorem}
\noindent\textbf{Adaptive Algorithms Complexity\\}
The method above is near-optimal with respect to non-adaptive algorithms. An important question is how much can we gain from adaptiveness in the test selection process.
We now show a lower bound of $\Omega\parentheses{{u+d\choose d}+d\logp{n/d}}$ for adaptive algorithms. The gap from the $\parameterizedCFFSize{d}{u+1}={{u+d\choose d}(u+d)^{O(1)}\log n}$ query complexity of our non-adaptive algorithm is left as future work. 
\ifdefined\TACASVersion
\begin{theorem}\label{thm:adaptiveMaxInstallation}(*)
	Assume that the repository contains at most $d$ defects and $u$ unknown dependencies.
	Any \textbf{adaptive} algorithm for \maxProblem{} must make $\Omega\parentheses{{u+d\choose d}+d\logp{n/d}}$ queries.
\end{theorem}
\else
\begin{theorem}\label{thm:adaptiveMaxInstallation}
	Assume that the repository contains at most $d$ defects and $u$ unknown dependencies.
	Any \textbf{adaptive} algorithm for \maxProblem{} must make $\Omega\parentheses{{u+d\choose d}+d\logp{n/d}}$ queries.
\end{theorem}
\begin{proof}
	First, note that a $d\logp{n/d}$ lower bound immediately follows from the group testing lower bound (as the case of $K=U=\emptyset$ degenerates to group testing).
	Fix an arbitrary package subset $P'\subseteq P$ of size $|P'|=u+d$. Consider an algorithm that makes at most ${u+d\choose d}-2$ queries.
	Clearly, there are ${u+d\choose d}$ subsets of size $u$ of $P'$.
	Thus, there exists two disjoint-subsets pairs $A_1,A_2\subseteq P'$ and $B_1,B_2\subseteq P'$ such that $|A_1|=|B_1|=d$ and $|A_2|=|B_2|=u$ (also, $A_1\cap A_2=B_1\cap B_2=\emptyset$) and that no tested installation contains $A_2$ but none of $A_1$ and no test includes $B_2$ but none of $B_1$.
	Consider the following scenarios:
	\begin{enumerate}
		\item The set of defects is $D=A_1$ and all of the packages in $A_2$ depend on each other (i.e., there is a cycle that contains all of $A_2$ in $(P, U)$).
		\item The set of defects is $D=B_1$ and all of the packages in $B_2$ depend on
		each other (there is a cycle that contains all of $B_2$ in $(P, U)$).
	\end{enumerate} 
	Notice that in case I, any installation that contains $A_2$ but none of the packages in $A_1$ passes and similarly for case II and $B_1, B_2$. As no such installations were tested by the algorithm, every test that contains at least one package of $P'$ fails, regardless of the actual scenario. Thus, the algorithm cannot determine whether $A_2$ or $B_2$ belongs to the maximal installation and thus fails to solve the~problem.
\end{proof}
\fi
\section{Learning Conflicts when All Dependencies are Known}

\ifdefined\TACASVersion
Previously, we assumed that the repository contained no conflicts, and
identified the defects. In this section, we return to the case where all
dependencies are known, but now the repository may have unreported conflicts.
We show here that learning conflicts is ``hard'', in the sense that even
when all dependencies are known and no defects exist -- identifying the exact conflicts requires a linear number of queries.

\begin{theorem}\label{thm:learningAStrongConflict}
Assume that all dependencies are known and no defects exist.
Any non-adaptive algorithm that solves  \learnAll{} must make at least $n-1$
queries. Note that this holds even when the repository \mbox{may have only up to $1$
conflict.}
\end{theorem}
\begin{proof}
	Consider the repository $P\triangleq\set{p,a_1,a_2,\ldots,a_{n-1}}$ and the 
	dependencies $K\triangleq\{(a_{i+1},a_i)\mid$ $ i \in\set{1,2,\ldots n-2}\}$.
	If the algorithm makes $n-2$ queries or less, then there exists some
	$\ell\in\set{2,\ldots,n-1}$, such that the algorithm does not query
	$(\set{p}\cup\set{a_i\mid i<\ell})$. In this case, the feedback for all 
	the queries is the same in both $C=\set{\set{p,a_\ell}}$ case and
	$C=\set{\set{p,a_{\ell-1}}}$ case. Thus, the algorithm cannot determine 
	which package is $p$ conflicting with and it fails to solve \learnAll{}. 
	Similarly, if the algorithm does not attempt to install $\set{p,a_1}$, it 
	cannot distinguish between the case where the packages are in conflict with
	each other and the case where no conflicts exist.
\end{proof}

In order to ``regain'' the logarithmic query complexity, we consider a weaker notion of conflicts. For convenience, we also define the \emph{weak dependency} notion. 
\begin{definition}
	Given two packages $p,q\in P$, we say that $p$ \textbf{weakly
	depends}\ignore{\footnote{When all dependencies are known the weak dependencies
	of $v$ are simply $c(v)$.}} on $q$ if there exists no successful installation
	that includes $p$ but not $q$. Also, $p$ \textbf{weakly conflicts} with $q$ if no successful installation includes both $p$ and $q$.
\end{definition}

Armed with the relaxed definition, we now analyze the query complexity of algorithms given bounds on the number of defects and weak conflicts.
We prove a lower bound for \learnAll{}. Upper bounds for the adaptive and non-adaptive cases are proven in~\cite{fullVersion}.

\begin{theorem}\label{thm:learnAllLBWithConflictsKnownDependenciesLB}
	Assume that the repository contains at most $d$ defects and $c$ unknown \textbf{weak} conflicts.
	Any \textbf{non-adaptive} algorithm that solves \learnAll{} must make at least $\parameterizedCFFSize{d+c-1}{2}$ queries.
\end{theorem}
\begin{proof}
	Assume by contradiction that the algorithm makes less than $\parameterizedCFFSize{d+c-1}{2}$ queries. 
	This means that there exist a set $S_1=\set{x_1,\ldots,x_{d}}\cup\set{b_1,\ldots,b_{c-1}}$ and  packages $p,q$ such 
	that some tested installation contains $p$ and $q$ but none of $S_1$'s members. 
	We will construct an input scenario for which \learnAll{} cannot be solved
	correctly by this algorithm. Consider the scenario where
	$x_1,\ldots,x_{d}$ are defective packages and $b_1,\ldots,b_{c-1}$ conflict
	with $p$.
	This implies that no installation that contains $p$ \textbf{and} $q$ is tested. Thus, the algorithm cannot determine whether $p$ conflicts with $q$ or not. 
	The setting is depicted in Figure~\ref{fig:learnAllLBWithConflictsKnownDependenciesLBV2}.
		Observe that the number of weak conflicts is at most $c$ as required and there are $d$ defects.
	Thus, the algorithm fails to solve \learnAll{}.
	\begin{figure}[]
		\centering
		\includegraphics[width =
		0.7\columnwidth]{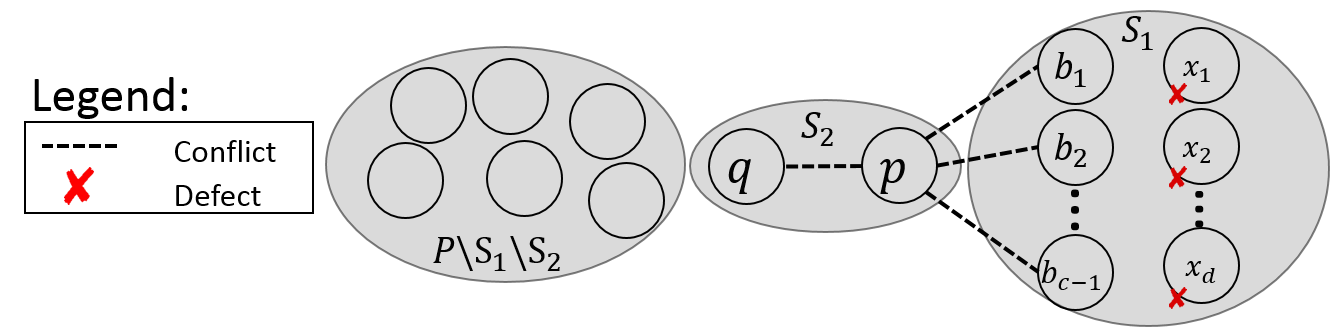}
		
		\caption{If the algorithm makes less than $\parameterizedCFFSize{d+c-1}{2}$
		queries, then there exists $p,q,\set{b_i}\set{x_i}$ such that no installation that contains both $p$ and $q$, but exclude $\set{b_i}\set{x_i}$, is tested; thus
		the algorithm \mbox{can not determine if they conflict.}
	 }
		\label{fig:learnAllLBWithConflictsKnownDependenciesLBV2}
	\end{figure}
\end{proof}

\else
Previously, we assumed that the repository contained no conflicts, and
identified the defects. In this section, we return to the case where all
dependencies are known, but now the repository may have unreported conflicts.
We start by showing that learning conflicts is ``hard'', in the sense that even
when all dependencies are known and no defects exist -- identifying the exact conflicts requires a linear number of queries.
\begin{theorem}\label{thm:learningAStrongConflict}
Assume that all dependencies are known and no defects exist.
Any non-adaptive algorithm that solves  \learnAll{} must make at least $n-1$
queries. Note that this holds even when the repository \mbox{may have only up to $1$
conflict.}
\end{theorem}
\begin{proof}
	Consider the repository $P\triangleq\set{p,a_1,a_2,\ldots,a_{n-1}}$ and the 
	dependencies $K\triangleq\{(a_{i+1},a_i)\mid$ $ i \in\set{1,2,\ldots n-2}\}$.
	If the algorithm makes $n-2$ queries or less, then there exists some
	$\ell\in\set{2,\ldots,n-1}$, such that the algorithm does not query
	$(\set{p}\cup\set{a_i\mid i<\ell})$. In this case, the feedback for all 
	the queries is the same in both $C=\set{\set{p,a_\ell}}$ case and
	$C=\set{\set{p,a_{\ell-1}}}$ case. Thus, the algorithm cannot determine 
	which package is $p$ conflicting with and it fails to solve \learnAll{}. 
	Similarly, if the algorithm does not attempt to install $\set{p,a_1}$, it 
	cannot distinguish between the case where the packages are in conflict with
	each other and the case where no conflicts exist.
\end{proof}

In order to ``regain'' the logarithmic query complexity, we consider a weaker notion of conflicts. For convenience, we also define the \emph{weak dependency} notion. 
Intuitively, $p$ weakly depends on $q$ if there exist packages 
$\set{q_1, q_2, ..., q_i}$,  such that $(p, q_1), (q_1, q_2),..., (q_i, q)$.
\begin{definition}
	Given two packages $p,q\in P$, we say that $p$ \textbf{weakly
	depends}\ignore{\footnote{When all dependencies are known the weak dependencies
	of $v$ are simply $c(v)$.}} on $q$ if there exists no successful installation
	that includes $p$ but not $q$. Also, $p$ \textbf{weakly conflicts} with $q$ if no successful installation includes both $p$ and $q$.
\end{definition}

Armed with the relaxed definition, we now analyze the query complexity of algorithms given bounds on the number of defects and weak conflicts.
We start with a lower bound for \learnAll{}.

\begin{theorem}\label{thm:learnAllLBWithConflictsKnownDependenciesLB}
	Assume that the repository contains at most $d$ defects and $c$ unknown \textbf{weak} conflicts.
	Any \textbf{non-adaptive} algorithm that solves \learnAll{} must make at least $\parameterizedCFFSize{d+c-1}{2}$ queries.
\end{theorem}
\begin{proof}
	Assume by contradiction that the algorithm makes less than $\parameterizedCFFSize{d+c-1}{2}$ queries. 
	This means that there exist a set $S_1=\set{x_1,\ldots,x_{d}}\cup\set{b_1,\ldots,b_{c-1}}$ and  packages $p,q$ such 
	that some tested installation contains $p$ and $q$ but none of $S_1$'s members. 
	We will construct an input scenario for which \learnAll{} cannot be solved
	correctly by this algorithm. Consider the scenario where
	$x_1,\ldots,x_{d}$ are defective packages and $b_1,\ldots,b_{c-1}$ conflict
	with $p$.
	This implies that no installation that contains $p$ \textbf{and} $q$ is tested. Thus, the algorithm cannot determine whether $p$ conflicts with $q$ or not. 

	The setting is depicted in Figure~\ref{fig:learnAllLBWithConflictsKnownDependenciesLBV2}.
		Observe that the number of weak conflicts is at most $c$ as required and there are $d$ defects.
	Thus, the algorithm fails to solve \learnAll{}.
	\begin{figure}[]
		\centering
		\includegraphics[width =
		0.6\columnwidth]{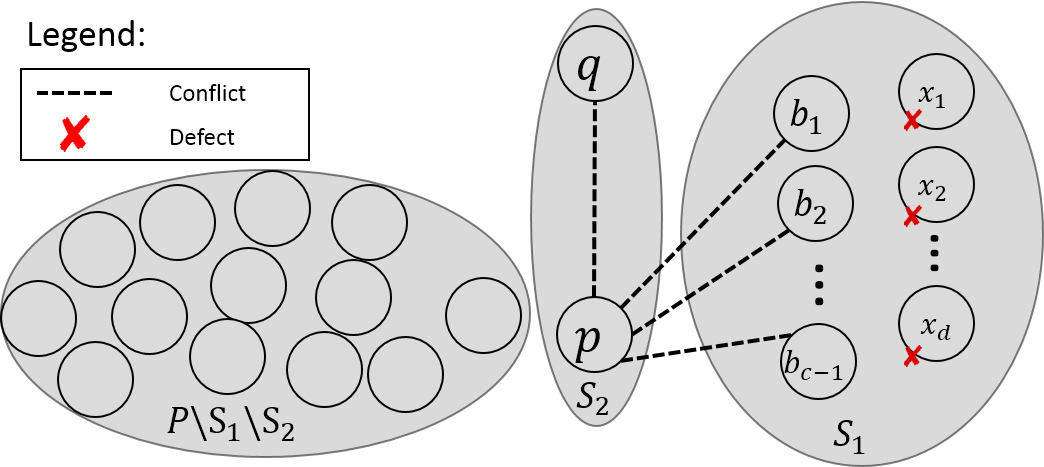}
		
		\caption{If the algorithm makes less than $\parameterizedCFFSize{d+c-1}{2}$
		queries, then there no installation that contains both p and q is tested, thus
		the algorithm can not determine if they are in conflict.
	 }
		\label{fig:learnAllLBWithConflictsKnownDependenciesLBV2}
	\end{figure}
\end{proof}

We proceed with a non-adaptive algorithm for the \learnAll{} problem. 
The test selection is similar to that of the previous section. Namely, we construct a \parameterizedCFF{d+c}{2} \CFFLetter{} and propagate the known dependencies so that the tests are $\mathcal T \triangleq \set{c(v)\mid v\in \CFFLetter}$ (see Section~\ref{sec:unknownDeps} for a formal definition of $c(v)$). 
Following are lemmas that show that from the feedback we can \mbox{infer all~constraints.}

\begin{lemma}\label{lem:defectIdentificationWithConflictsKnownDependencies}
	A package $p\in P$ is defective if and only if $S(p)=\emptyset$.
\end{lemma}
\begin{proof}
	The claim is similar to Lemma~\ref{lem:defectiveness} except that now we may
	have conflicts, and all the dependencies are known. Once again, if $p$ is
	defective, then clearly $S(p)=\emptyset$, and our goal is to show the converse.
	Given a non-defective $p$, we construct sets $S_1, S_2$ of the required sizes
	such that a test containing all of $S_2$ and none of $S_1$ succeeds, thus
	showing that $S(p) \neq \emptyset$.
	Intuitively, we satisfy each conflict by excluding (having in $S_1$) one of its packages along with all those that depend on it. 
	As for $S_2$, all we need is to include $p$ as the dependency propagation will ensure that all its prerequisites are installed as well.
	The setting is illustrated in Figure~\ref{fig:defectIdentificationIllustrationKnownDependencies}.
	
	Next, notice that if $p$ weakly conflicts with a prerequisite of itself, 
	then $p$ cannot be successfully installed and is thus a defect. Hence, we
	hereafter assume that $p$ does not conflict with any of its prerequisites.
	We now formally define a pair of disjoint 
	sets $S_1, S_2$ such that the corresponding test contains $p$ and passes.
	Here, $S_2 = \set{p}$.
	Next, we consider an arbitrary order $\succ$ on $P$ that will allow us to resolve the conflict constraints. 
	We define two package sets as follows:
	\begin{itemize}
		\item $X_1\triangleq \set{q\mid \exists p':\set{q,p'}\in C \wedge\mbox{$p$ weakly depends on $p'$}}$.
		\item $X_2\triangleq \set{q\mid \exists p':(\set{q,p'}\in C) \wedge\mbox{($p$ does not weakly depend on $q$)}\wedge(q\succ p')}$.
	\end{itemize} 
	
	Intuitively, $X_1$ contains all packages that conflict with a prerequisite of $p$; $X_2$'s packages are those that have no relation to $p$ and have an unknown conflict with another package with a lower index (according to $\succ$).
	If we make sure that we have a test that installs all of $p$'s prerequisites
	and excludes all packages in $E\triangleq X_1\cup X_2$, it will pass if $p$ is not defective. 
	However, due to the propagation of the known dependencies, it is not enough to 
	include \textit{just} $D\cup E$ in $S_1$. That is, if a package $q\in E$ is a member 
	of $S_1$, but a package that depends on it has `$1$` in the corresponding
	vector $v$ in the CFF, the propagation-result $c(v)$ will include $q$ as well.
	We circumvent this issue by adding to $S_1$ \textit{all} packages that weakly depend on $E$, i.e., 
	we set $S_1\triangleq D\cup \set{q\mid \exists q'\in E: \mbox{$q$ weakly depends on $q'$}}$. 
	Since every package in $S_1$ has an unknown weak conflict or defect we have that $|S_1|\le d+c$. 
	Figure~\ref{fig:defectIdentificationIllustrationKnownDependencies} illustrates the sets $S_1,S_2$ that allow the corresponding test $c(v)$ to pass if $p$ is not defective.
	\begin{figure}[h]
		\centering
		\includegraphics[width =
		0.8\columnwidth]{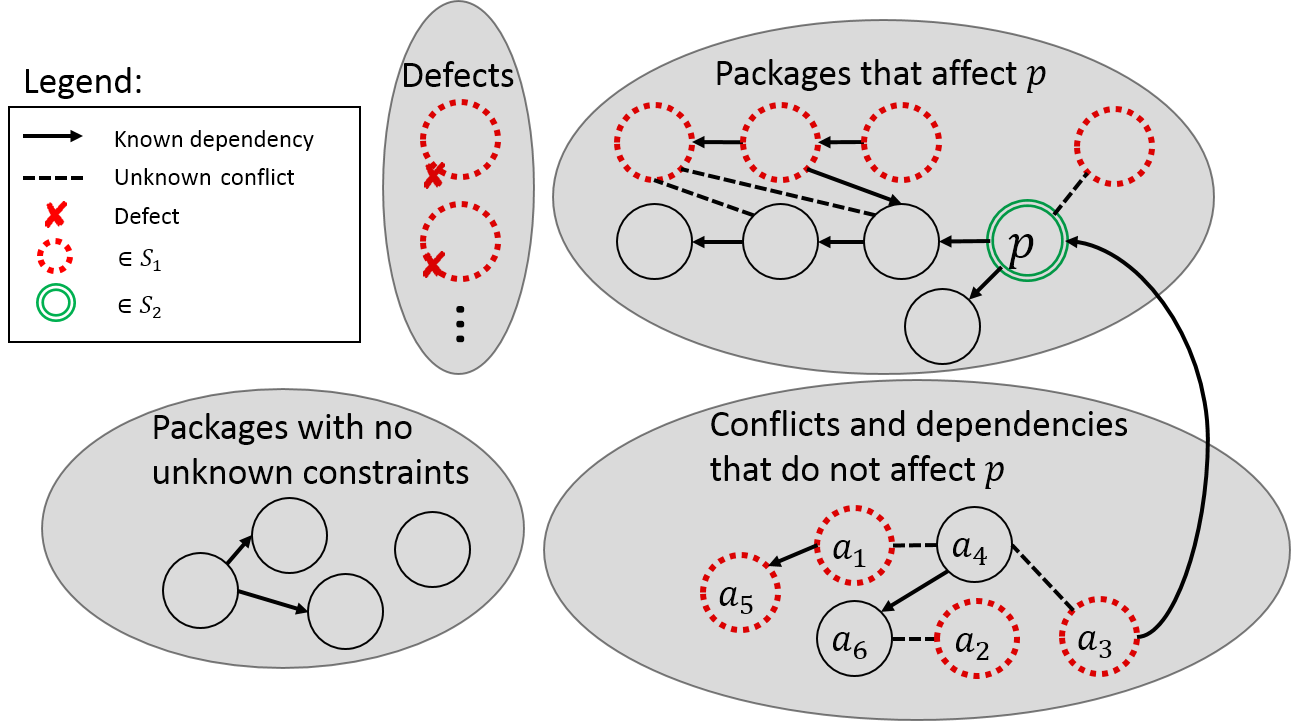}
		
		\caption{An example of a $S_1,S_2$ whose corresponding test 
		contains $p$ and passes. 
		$S_2$ contains $p$, while $S_1$ contains all packages that weakly 
		conflict with $p$ and a package from each conflict, 
		along with all those that weakly depend on it. The latter are selected based
		on the order $\succ$, such that $a_1$ is selected to be in $S_1$ first, along
		with its prerequisite $a_5$, followed by $a_2$ and $a_3$.
		}
		\label{fig:defectIdentificationIllustrationKnownDependencies}
	\end{figure}
	Since we have $|S_1|\le d+c$ and $|S_2|=1<2$ there exists a vector
	$v\in\CFFLetter$, such that $S_1\cap v=\emptyset$ and $S_2\subseteq v$. By the
	construction of $S_1$ and $S_2$ we are guaranteed that the test includes $p$ and satisfies all conflict and dependency constrains.
\end{proof}

As mentioned, in the presence of conflicts, identifying all defects is not enough even for solving \maxProblem{}. We therefore show that our algorithm can also learn the conflicts themselves. Intuitively, if two packages weakly conflict, there exists no installation that contains both.
Thus, we need to show that for every pair of \textbf{non-conflicting} packages
$p,q$ our algorithm has a witness -- a successful test that contains both.

\begin{lemma}\label{lem:learnConflictsForKnownDependencies}
	Packages $p,q$ weakly conflict if and only if $S(p)\cap S(q)=\emptyset$.
\end{lemma}
\begin{proof}
	Notice that if $p,q$ weakly conflict then no test that contains both can pass 
	and $S(p)\cap S(q)=\emptyset$. In the remainder of the proof, we assume 
	that the two do not conflict and show that the algorithm tries a 
	successful installation that contains both. Similarly to the proof of 
	Lemma~\ref{lem:defectIdentificationWithConflictsKnownDependencies}, we 
	define $S_2\triangleq\set{p,q}$; the resulting test will include $p,q$ and all their prerequisites. 
	We also construct $S_1\triangleq D\cup \set{x\mid \exists q'\in E: \mbox{$x$ weakly depends on $q'$}}$ in a similar manner, 
	where $E\triangleq X_1\cup X_2$ and 
	\begin{itemize}
		\item $X_1\triangleq \set{x\mid \exists p':\set{x,p'}\in C \wedge\mbox{$p$
		\textbf{or} $q$ weakly depends on $p'$}}$
	\item $X_2\triangleq \set{x\mid \exists p':(\set{x,p'}\in C) \wedge\mbox{($p$
	\textbf{and} $q$ does not weakly depend on $x$)}\wedge(x\succ p')}$.
	\end{itemize}
	
	Observe that $|S_1|\le d + c$ and $|S_2|=2$. Thus, there exists 
	$v\in\CFFLetter$ such that $v\cap S_1=\emptyset$ and $S_2\subseteq v$.
	Since the resulting test $c(v)$ contains all of the prerequisites of $p,q$ and
	satisfies all other constraints, and since $p,q$ do \textbf{not} (weakly)
	conflict, this test~passes.\qedhere
\end{proof}
 
Putting the lemmas together, we conclude that our algorithm can identify all defects 
and conflicts.
The pseudocode for the algorithm is shown
in Algorithm~\ref{alg:conflicts_w_unknowns}.

\begin{algorithm}[H]
    \ifdefined\fullVersion
	    \footnotesize
	\else
	\scriptsize
	\fi
	\begin{algorithmic}[1]
		\Function{LearnAll}{$P,K,c,d$} \Comment{At most $c$ weak conflicts and $d$
		defects}
		\State $D \gets \emptyset$ \Comment{The set of identified defects}
		\State $C \gets \emptyset$ \Comment{The set of identified weak conflicts}
		\State $\CFFLetter \gets CFF(n,d+c,2)$ \Comment{CFF construction}
		\State $\mathcal T \gets \set{c(v)\mid v\in \CFFLetter}$ \Comment{Test vectors
		generated based on all known dependencies}
		\For {$p \in P$ }
			\State $S(p) \gets \emptyset$ \Comment{Successful installations that
			included $p$}
		\EndFor
		\For {$t \in \mathcal T$}
		\State $f(t) \gets Feedback(t)$ \Comment{Get feedback from oracle}
		\If {$f(t)$ = 1} \Comment{Successful installation}
		\For {$p$ tested as part of $t$}
			\State $S(p) \gets S(p) \cup t$
		\EndFor        
		\EndIf
		\EndFor
		
		\For {$p \in P$ }
			\If {$S(p) = \emptyset$} \Comment{No successful installation exists for $p$}
			\State $D \gets D \cup p$
			\EndIf
			\For {$q \in P, q \neq p$}
			\If {$S(p) \cap S(q) = \emptyset$} 
			\ifdefined\TACASVersion
			\Comment{\scriptsize No successful installation exists for $p$ and $q$ together} 
			\else
			\Comment{No successful installation exists for $p$ and $q$ together} 
			\fi
			\State $C \gets C \cup \{p,q\}$
			\EndIf
			\EndFor
		\EndFor

		\State\Return $D,C$
		\EndFunction
		
	\end{algorithmic}
	\normalsize
	\ifdefined\TACASVersion
	\caption{\small \learnAll{} with defects and conflicts, all dependencies are known}
	\else
	\caption{\learnAll{} with defects and conflicts, all dependencies are known}
	\fi
	\label{alg:conflicts_w_unknowns}
\end{algorithm}

\begin{theorem}\label{thm:learnAllLBWithConflictsUBKnownDependencies}
	Assume that the repository contains at most $d$ defects and $c$ weak conflicts.
	There exists a \textbf{non-adaptive} algorithm that solves \learnAll{} using
	$\parameterizedCFFSize{d+c}{2}=(d+c)^{3+o(1)}\log n$ queries.
\end{theorem}
\subsection{Using Adaptivity to Reduce Complexity}
The algorithm and lower bound presented above indicate a query dependency of $d^{3\pm o(1)}$ on the number of defects.
As the number of defects may be large, a natural question is whether this cubic dependency can be improved using adaptiveness. 
We now show that, indeed, adaptive algorithms have just a $d^{2+ o(1)}$ dependency on the number of defects. 

\begin{theorem}\label{thm:adaptiveConflictsWithKnownDependencies}
	Assume that the repository contains at most $d$ defects and $c$ weak conflicts.
	There exists an \textbf{adaptive} algorithm that solves \learnAll{} using $\parameterizedCFFSize{d+c}{1}+\parameterizedCFFSize{c}{2}$ queries.
\end{theorem}
\begin{proof}
	The idea behind our algorithm is first to identify the defects, remove them from the repository and then follow the non-adaptive procedure when no defects exist.
	Recall that in the proof of
	Lemma~\ref{lem:defectIdentificationWithConflictsKnownDependencies}, the size of
	$S_2$ was just one (only the package for which we wish to assert
	defectiveness).
	Hence, Lemma~\ref{lem:defectIdentificationWithConflictsKnownDependencies} also
	holds for \parameterizedCFFSize{d+c}{1}. Therefore, by taking a
	\parameterizedCFF{d+c}{1} and propagating the known dependencies, we can find all defects.
	The adaptivity then allows us to compute the subsequent \parameterizedCFF{c}{2}
	only on the non-defective packages and learn all conflicts.
\end{proof}

\fi
\ifdefined\fullVersion
\section{Learning with Unknown Defects, Dependencies, and Conflicts}
\else
\section{\mbox{Learning Unknown Defects, Dependencies, and Conflicts}}
\fi

\label{sec:all} 
Heretofore, we considered cases where either unknown dependencies exist or unknown conflicts exist, but not both.
In this section, we discuss the query complexity of the most difficult scenario in which the repository may have defects, hidden dependencies, and unknown conflicts.

We start by proving that learning unknown dependencies requires a linear number
of queries if the number of known dependencies is not bounded, even if no
conflicts or defects exist. 
\ifdefined\fullVersion
This will lead us to bound the \mbox{overall number of dependencies.}
\fi

\ignore{
\ifdefined\TACASVersion
\vspace*{-1mm}
\begin{theorem}\label{thm:learningADependency}(*)
	Assume that no defects or conflicts exist.
	Any non-adaptive algorithm that solves  \learnAll{} must make at least $n-1$
	queries. Note that this holds even when the repository may have only \textit{up
	to} 1 unknown dependency.
\end{theorem}
\vspace*{-1mm}
\else
}
\begin{theorem}\label{thm:learningADependency}
	Assume that no defects or conflicts exist.
	Any non-adaptive algorithm that solves  \learnAll{} must make at least $n-1$
	queries. Note that this holds even when the repository may have only \textit{up
	to} 1 unknown dependency.
\end{theorem}
\begin{proof}
	Denote $P\triangleq\set{a_1,\ldots,a_n}$ and consider the directed path graph 
	given by $K\triangleq\{(a_i,a_{i-1})\mid i\in\set{2,\ldots,n}\}$. 
	Assume by contradiction that there exists an algorithm that makes 
	less than $n$ queries. In such case, there exists some
	$i\in\set{1,\ldots,n-1}$, such that the installation $\set{1,\ldots, i}$ was not tested. Next, consider two
	possible scenarios -- $U=\set{(a_i,a_n)}$ and $U=\set{(a_{i+1},a_n)}$. That
	is, there is an undocumented dependency of either $a_i$ or $a_{i+1}$ on $a_n$, as illustrated in Figure~\ref{fig:learningDependencyLB}.
		The \emph{only} way to distinguish between the two cases is to try
	the installation $\set{1,\ldots i}$, as all other installations will result in
	the same outcome for both scenarios. Thus, the algorithm must try at least $n-1$ queries.
	\begin{figure}[b]
		\centering
		\includegraphics[width = 0.6\columnwidth]{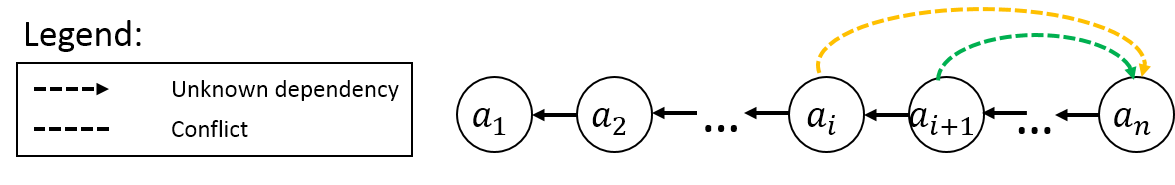}
		
		\caption{An algorithm that tries to learn a single unknown dependency must differentiate between the case where $a_n$ depends on $a_i$ and the case where it depends on $a_{i+1}$. This is only possible by trying the installation $\set{1,\ldots i}$.
		}
		\label{fig:learningDependencyLB}
	\end{figure}
\end{proof}
To circumvent the problem above, we consider a bound on the overall number of dependencies in the repository. 
That is, we henceforward assume that there are at most $u$ unknown dependencies. 
Formally, we set $K=\emptyset$ and provide algorithms 
that learn all constraints without prior knowledge of the dependencies. 
An interesting implication of this assumption is that our upper bound now
depends on the number of conflicts and not the, potentially larger, number of weak conflicts. 
\ifdefined\fullVersion
Thus, the gap in query complexity might not be as large as it seems.
We start with lower bounds where the first considers the maximal \mbox{installation problem.}
\else
We start with a lower bound for \maxProblem{}.
\fi
\ifdefined\TACASVersion
\begin{theorem}\label{thm:maxLowerBoundWithConflicts}(*)
	Assume that the repository contains at most $d$ defects, $u$ unknown 
	dependencies, and $c$ unknown \textbf{weak} conflicts.
	For the case that $c\le 2u$, any \textbf{non-adaptive} algorithm that solves
	\maxProblem{} must use at least $\parameterizedCFFSize{d+\ceil{c/2}-1}{u+1}$ queries.
\end{theorem}
\else
\begin{theorem}\label{thm:maxLowerBoundWithConflicts}
	Assume that the repository contains at most $d$ defects, $u$ unknown 
	dependencies, and $c$ unknown \textbf{weak} conflicts.
	For the case that $c\le 2u$, any \textbf{non-adaptive} algorithm that solves
	\maxProblem{} must use at least $\parameterizedCFFSize{d+\ceil{c/2}-1}{u+1}$ queries.
\end{theorem}
\begin{proof}
	The setting for the proof is similar to that of
	Theorem~\ref{thm:unknownDependenciesLB} except that there are $\ceil{c/2}-1$
	packages, each conflicting with one of $a_1,\ldots,
	a_{\ceil{c/2}-1}$, therefore they all indirectly conflict with $p$. Another
	package directly conflicts with $p$.
	Formally, if we assume that the algorithm tests less than $\parameterizedCFFSize{d+\ceil{c/2}-1}{u+1}$ installations, 
	there exists a pair of package sets $S_1\triangleq\set{b_1,\ldots,b_{\ceil{c/2}}}\cup\set{x_1,\ldots,x_{d-1}}$ and 
	$S_2\triangleq\set{a_1,\ldots,a_u}\cup\set{p}$ such that no test contains all of $S_2$ but none of the packages in $S_1$. 
	In this case, we prove that the algorithm cannot determine whether $p$ is defective in the following scenario.
	Consider the dependencies $(p,a_1),(p,a_2),\ldots,(p,a_u)$, defects
	$x_1,\ldots,x_{d-1}$ and conflicts $\set{b_1,a_1}, \ldots,\set{b_{\ceil{c/2}-1},a_{\ceil{c/2}-1}},\allowbreak\set{b_{\ceil{c/2}},p}$. 
	The setting is illustrated in Figure~\ref{fig:maxLowerBoundWithConflicts}.
	
	\begin{figure}[]
		\centering
		\includegraphics[width = 0.6\columnwidth]{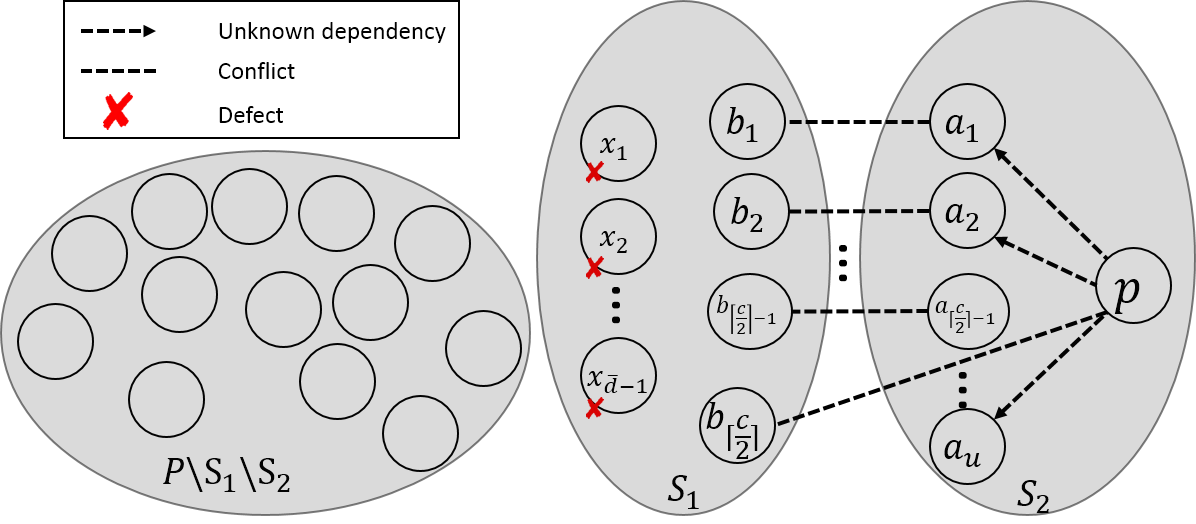}
		
		\caption{If the algorithm makes less than 
		$\parameterizedCFFSize{d+\ceil{c/2}-1}{u+1}$
		queries, it can not determine if $p$ is defective.
		}
		\label{fig:maxLowerBoundWithConflicts}
	\end{figure}
	
	Notice that the number of dependencies is $u$ and that the number of 
	weak conflicts is $c$ (as $\set{b_i,a_i}$ and $\set{b_i,p}$ are 
	weak conflicts for $i\in\set{1,\ldots,\ceil{c/2}-1}$). As in Theorem~\ref{thm:unknownDependenciesLB}, the algorithm 
	has no way to distinguish between the case where $p$ is defective and the case that it is not, 
	as no test that contains $p$ has passed. Finally, notice that the maximal 
	installation contains $a_1,\ldots, a_u$ and potentially $p$, and thus the algorithm fails to solve the problem.
\end{proof}
\fi

\ifdefined\TACASVersion
We also prove a stronger bound for the \learnAll{} problem.
\else
Using a slightly different construction, we prove a stronger bound for the \learnAll{} problem.
\fi
%
%
%

\ignore{
\ifdefined\TACASVersion
\begin{theorem}\label{thm:learnAllLBWithConflictsV2}(*)
	Assume that the repository contains at most $d$ defects, $u$ unknown
	dependencies and $c$ unknown \textbf{weak} conflicts.
	Any \textbf{non-adaptive} algorithm that solves \learnAll{} must use at least $\parameterizedCFFSize{d+c-1}{u+2}$ queries.
\end{theorem}
\else
}

\begin{theorem}\label{thm:learnAllLBWithConflictsV2}
	Assume that the repository contains at most $d$ defects, $u$ unknown
	dependencies and $c$ unknown \textbf{weak} conflicts.
	Any \textbf{non-adaptive} algorithm that solves \learnAll{} must use at least $\parameterizedCFFSize{d+c-1}{u+2}$ queries.
\end{theorem}
\begin{proof}
	Unlike the \maxProblem{} 
	bound discussed earlier, an algorithm solving
	\learnAll{} must determine for every $p,q$ whether the two conflict.
	This allows us to consider the case where $c-1$ packages $b_1,\ldots,b_{c-1}$
	directly conflict with $p$ that also has $u$ unknown dependencies on
	$a_1,\ldots,a_u$, in addition to $x_1,\ldots,x_d$ defective packages.
	The setting is depicted in Figure~\ref{fig:learnAllLBWithConflictsV2}.
	
	\begin{figure}[]
		\centering
		\includegraphics[width = \columnwidth]{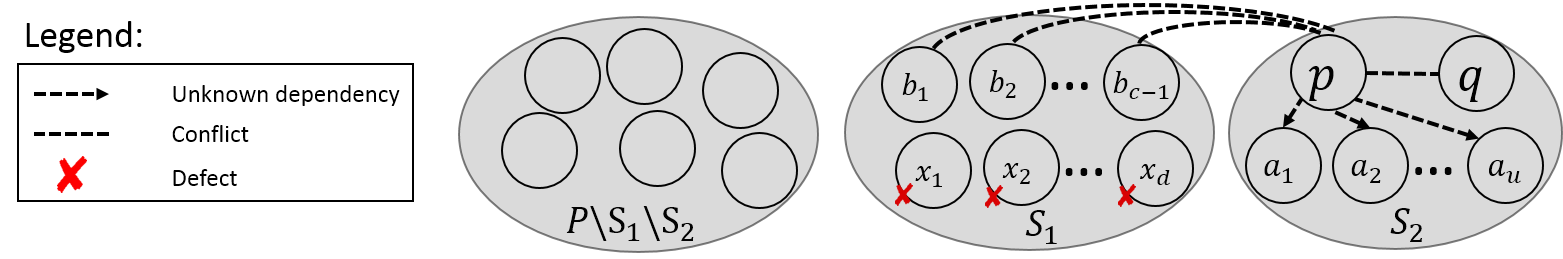}
		
	\caption{In the scenario $x_1,\ldots,x_d$ are defective and 
		$b_1,\ldots,b_{c-1}$ conflict with $p$ that also has prerequisites $a_1,\ldots,a_u$, the only way to determine if there is a conflict between $p$ and $q$ is to have a test that includes all of $S_2$ and excludes $S_1$.
		}
		\label{fig:learnAllLBWithConflictsV2}
	\end{figure}
	
	Observe that the number of unknown dependencies is $u$, the number of defects is $d$, and the number of weak conflicts is at most $c$, as required.
	Since no installation that contains $\set{a_1,\ldots,a_u}\cup\set{p,q}$ and
	excludes $\set{b_1,\ldots,b_{c-1}}\cup\set{x_1,\ldots,x_d}$ is tested, the
	algorithm cannot determine whether $p$ conflicts with $q$ or not, as the
	feedback for all other tests is identical for the two cases.
\end{proof}


We proceed with a non-adaptive algorithm for the \learnAll{} problem. 
The test selection is similar to that of the previous section, except that now we have no known dependencies to propagate. Namely, we construct a \parameterizedCFF{d+c+u+1}{u+2} \CFFLetter{} and the tests are simply $\mathcal T\triangleq \CFFLetter{}$.
\ifdefined\TACASVersion
In the full version of this paper~\cite{fullVersion}, we prove lemmas that show that (i) A package $p\in P$ is defective if and only if $S(p)=\emptyset$. (ii) Packages $p,q$ weakly conflict if and only if $S(p)\cap S(q)=\emptyset$. (iii) Package $p$ weakly depends on $q$ if and only if $S(p)\subseteq S(q)$.
\else
Following are lemmas that show that from the feedback of the tests 
we can infer all constraints. 
\ifdefined\TACASVersion
\begin{lemma}\label{lem:defectIdentificationWithConflicts}(*)
	A package $p\in P$ is defective if and only if $S(p)=\emptyset$.
\end{lemma}
\else
\begin{lemma}\label{lem:defectIdentificationWithConflicts}
	A package $p\in P$ is defective if and only if $S(p)=\emptyset$.
\end{lemma}
\begin{proof}
	The difference between the model here and that of Lemma~\ref{lem:defectIdentificationWithConflictsKnownDependencies} is that we now allow unknown dependencies and assume that $K=\emptyset$.
	As before, a defective package $p$ yields $S(p)=\emptyset$; we show that if $p\notin D$ then we have a successful test that includes $p$ to witness that.
	Intuitively, we satisfy each conflict by excluding (having in $S_1$) one of its 
	packages along with all those that depend on it. Unlike before, we split our unknown dependency resolution to cases. We wish to have all packages that $p$ 
	weakly depends on active (i.e., in $S_2$). All other packages that have an
	unknown dependency are excluded (placed in $S_1$), and so are all the packages
	that depend on them.
	The setting is illustrated in Figure~\ref{fig:defectIdentificationIllustration}.
	
	%
	%
	Next, notice that if $p$ weakly depends on two packages $p',p''$ that conflict, then $p$ cannot be successfully installed and is thus a defect. Therefore, we hereafter assume that no two packages that $p$ weakly depends on conflict.
	
	Following are formal definitions of sets $S_1, S_2$ such that the test that
	contains $S_2$ and avoids $S_1$ succeeds, thus serving as a witness to the
	non-defectiveness of $p$.
	We define $S_2$ as $p$ and all of the the packages it depends on:
	$S_2\triangleq\set{p}\cup\set{p'\mid\mbox{$p$ has an \textbf{unknown}
	weak dependency on $p'$}}$.
	Once again, we consider some order $\succ$ on $P$ for resolving the conflicts. 
	We define three package sets as follows:
	\begin{itemize}
		\item $X_1\triangleq \set{q\mid \mbox{($p$ does not weakly depend on $q$)}\wedge(\exists p'\in P: (p',q)\in U)}$.
		\item $X_2\triangleq \set{q\mid \exists p':(\set{q,p'}\in C) \wedge(\mbox{$p$ weekly depends on $p'$})}$.
		\item $X_3\triangleq \set{q\mid \exists p':(\set{q,p'}\in C) \wedge\mbox{($p$ does not weekly depends on $q$)}\wedge(q\succ p')}$.
	\end{itemize} 
	
	Intuitively, $X_1$ is the set of packages that have an unknown dependency and are not a prerequisite to $p$. $X_2$ contains all packages that conflict with a prerequisite of $p$. $X_3$'s packages are those that have no relation to $p$ and has an unknown conflict with another package with a lower index according to $\succ$.
	If we make sure that we have a test that installs all of $p$'s prerequisites and excludes all packages in $X_1\cup X_2\cup X_3$, it will pass unless $p$ is defective. 
	Thus, we define $S_1\triangleq D\cup X_1\cup X_2\cup X_3$.
	\begin{figure}[]
		\centering
		\includegraphics[width =
		0.75\columnwidth]{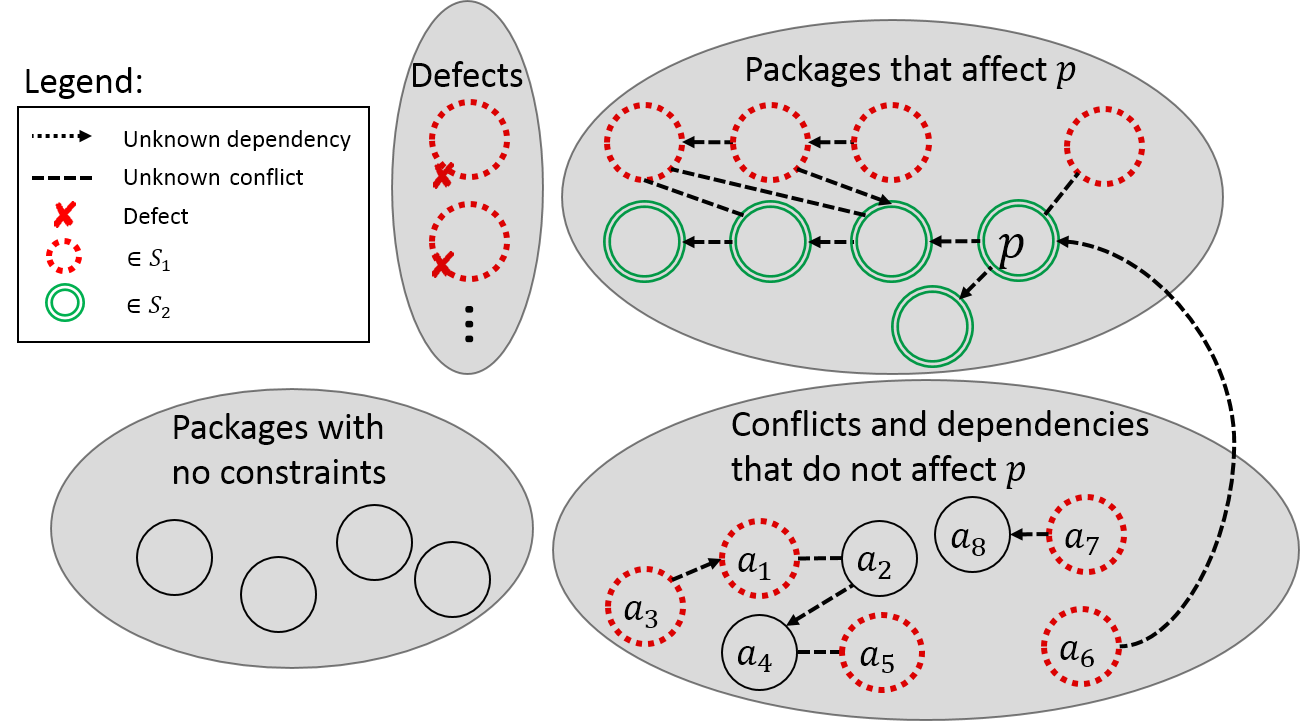}
		
		\caption{An example of a $S_1,S_2$ whose corresponding test contains 
		$p$ and passes. $S_2$ contains $p$ and all of its prerequisites. $S_1$
		contains all packages that weakly conflict with $p$, all packages $q$ that are not a 
		prerequisite to $p$ and have an unknown dependency, and a package from each
		conflict, chosen according to $\succ$.
		}
		\label{fig:defectIdentificationIllustration}
	\end{figure}
	Observe that every package in $S_1$ can be uniquely associated with a defect, conflict or dependency.
	Thus, we have $|S_1|\le d+c+u<d+c+u+1$ and $S_2\le u+1<u+2$ and hence there exists a vector $v\in\CFFLetter$ (and subsequently, a test) such that $S_1\cap v=\emptyset$ and $S_2\subseteq v$. By the construction of $S_1$ and $S_2$ we are guaranteed that the test includes $p$ and satisfies all conflict and dependency constrains.
\end{proof}
\fi

As mentioned, in the presence of conflicts, identifying all defects is not enough even for solving \maxProblem{}. We therefore show that our algorithm can also identify the dependencies and conflicts. 
As in the previous section, we show that if two packages do not conflict, 
then there exists a successful test that installs both and serves as a witness. 
\ifdefined\TACASVersion
\begin{lemma}\label{lem:learningConflictsWithUnknownDependencies}(*)
	Packages $p,q$ weakly conflict if and only if $S(p)\cap S(q)=\emptyset$.
\end{lemma}
\else
\begin{lemma}\label{lem:learningConflictsWithUnknownDependencies}
	Packages $p,q$ weakly conflict if and only if $S(p)\cap S(q)=\emptyset$.
\end{lemma}
\begin{proof}
	Similarly to Lemma~\ref{lem:learnConflictsForKnownDependencies}, a pair of conflicting packages $p,q$ cannot be included in a successful test. Here, we show that if no successful installation was tried by the algorithm, then such does not exist and the two packages conflict.
	In the remainder of the proof, we assume that the two do not conflict and show that the algorithm tries a successful installation that contains both. 
	We define $S_2\triangleq\set{p,q}\cup\set{p'\mid\mbox{$p$ or $q$ has an \textbf{unknown} weak dependency on $p'$}}$ to include $p,q$, and any of their unknown weak prerequisites. We also construct  $S_1\triangleq D\cup X_1\cup X_2\cup X_3$ in a similar manner, where
	\begin{itemize}
		\item $X_1\triangleq \set{x\mid \mbox{($p$ \textbf{and} $q$ do not weakly depend on $x$)}\wedge(\exists p'\in P: (p',x)\in U)}$.
		\item $X_2\triangleq \set{x\mid \exists p':(\set{x,p'}\in C) \wedge(\mbox{$p$ \textbf{or} $q$ weakly depend on $p'$)}}$.
		\item $X_3\triangleq \set{x\mid \exists p':(\set{x,p'}\in C) \wedge\mbox{($p$ \textbf{and} $q$ do not weakly depend on $x$)}\wedge(x\succ p')}$.
	\end{itemize}
	Observe that $|S_1|\le d + c + u<d+c+u+1$ and that $|S_2|\le u+2$. Therefore, there exists a test $v\in\mathcal T$ such that $v\cap S_1=\emptyset$ and $S_2\subseteq v$. The test contains all prerequisites of $p,q$ and satisfies all other constraints. Thus, it passes unless $p,q$ conflict.
\end{proof}
\fi

We are left with showing that the algorithm can also learn the unknown dependencies. 
If $p$ depends on $q$, no test that contains $p$ but not $q$ can pass. We show that if no successful installation with $p$ and 
not $q$ is tested by the algorithm, then such does not exist and $p$ depends on $q$. 
\ifdefined\TACASVersion
\begin{lemma}\label{lem:learningDependencies}(*)
	Package $p$ weakly depends on $q$ if and only if $S(p)\subseteq S(q)$.
\end{lemma}
\else
\begin{lemma}\label{lem:learningDependencies}
	Package $p$ weakly depends on $q$ if and only if $S(p)\subseteq S(q)$.
\end{lemma}
\begin{proof}
	As in the other lemmas, one direction here is straightforward -- if $p$ weakly depends on $q$ then $S(p)\subseteq S(q)$.
	We now construct sets $S_1,S_2$, such that the corresponding test will include $p$, exclude $q$, and pass unless $p$ weakly depends on $q$.
	This is achieved by excluding from the test $q$, and all the packages that weakly depend on it, including $p$ and all its prerequisites, and resolving all unrelated conflicts, dependencies, and defects.
	Formally we set $S_2\triangleq\set{p}\cup\set{p'\mid \mbox{$p$ has a weak dependency on $p'$}}$ to include $p$ and all its prerequisites.
	We define $S_1\triangleq D\cup \set{q}\cup X_1\cup X_2\cup X_3\cup X_4$, where:
	\begin{itemize}
		\item $X_1\triangleq\set{q'\mid \mbox{$q'$ has a weak dependency on $q$}}$
		\item $X_2\triangleq\set{x\mid (\mbox{$p$ does not weakly depend on
		$x$})\wedge(\exists y:(x,y)\in U)}$
		\item $X_3\triangleq \set{x\mid \exists p':(\set{x,p'}\in C) \wedge(\mbox{$p$ weekly depends on $p'$})}$.
		\item $X_4\triangleq \set{x\mid \exists p':(\set{x,p'}\in C) \wedge\mbox{($p$ does not
		 weakly depend on $x$)}\wedge(x\succ p')}$.
	\end{itemize}
	In other words, $X_1$ excludes from the test all packages that weakly depend on $q$; $X_2$ removes all packages with unknown dependencies that are not a prerequisite to $p$; $X_3$ contains all packages that conflict with a prerequisite of $p$; and $X_4$ is the set of packages, chosen according to $\succ$, that is used to resolve conflicts unrelated to $p$.
	An illustration of the setting is depicted in Figure~\ref{fig:learningDependencies}.
	\begin{figure}[]
	\centering
	\includegraphics[width = 0.8\columnwidth]{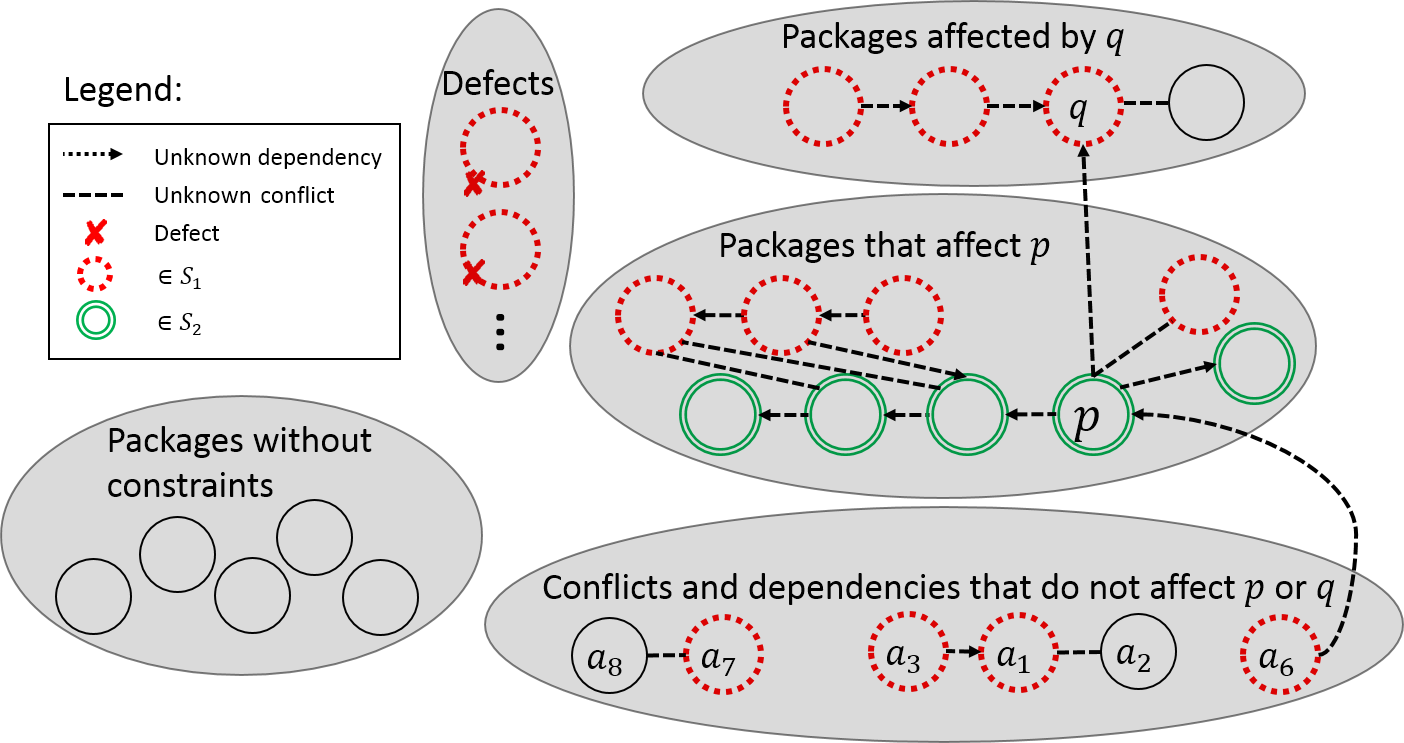}
	
	\caption{An example of a $S_1,S_2$ whose corresponding test 
		contains $p$ and passes. 
		$S_2$ contains $p$ and its prerequisites, while $S_1$ contains all packages
		that weakly conflict with $p$, $q$ and all those affecting q, and a
		package from each conflict, along with all those that weakly depend on it. 
	}
	\label{fig:learningDependencies}
	\end{figure}	
	Observe that since $|X_1|+|X_2|\le u$, $|X_3|+|X_4|\le c$, we have $|S_1|\le d + 1 + u + c$. Finally, since $|S_2|\le u+1<u+2$ we have that there exists a witness test that addresses all other constraints and passes unless $p$ weakly depends on $q$.
\end{proof}
\fi

\fi
As our algorithm learns all missing constraints using the tests in $\CFFLetter$, we conclude the~following.
\begin{corollary}\label{cor:learnAllLBWithConflictsUB}
	Assume that the repository contains at most $d$ defects, $u$ unknown dependencies and $c$ conflicts.
	There exists a \textbf{non-adaptive} algorithm for \learnAll{} that uses
	$\parameterizedCFFSize{d+c+u+1}{u+2}$~queries.
\end{corollary}

\ifdefined\TACASVersion
Finally, we discuss a test-design that allows us \mbox{to reduce the complexity.}
\begin{theorem}\label{thm:learnAllLBWithConflictsUB}(*)
	Assume that the repository contains at most $d$ defects, $u$ unknown dependencies and $c$ conflicts.
	There exists a \textbf{non-adaptive} algorithm for \learnAll{} that uses
	 $\sum_{i=0}^{u}\parameterizedCFFSize{d+c+(u-i)+1}{i+2}$ queries.
\end{theorem}
\else
Finally, we discuss a more careful test-design that allows us to reduce the number of queries needed.
\begin{theorem}\label{thm:learnAllLBWithConflictsUB}
	Assume that the repository contains at most $d$ defects, $u$ unknown dependencies and $c$ conflicts.
	There exists a \textbf{non-adaptive} algorithm for \learnAll{} that uses
	 $\sum_{i=0}^{u}\parameterizedCFFSize{d+c+(u-i)+1}{i+2}$ queries.
\end{theorem}
\begin{proof}
	Intuitively, our construction so far was wasteful in the sense that when a test needed to include $p$, we added all its prerequisites to $S_2$ and added every other package with an unknown dependency into $S_1$. While there can be at most $u$ prerequisites to $p$, and at most $p$ other packages with a dependency, we have not used the fact that their \textbf{sum} is also bounded by $u$. That is, assume that $p$ has $i\in\set{0,1,\ldots u}$ prerequisites; then we add $x$ packages to $S_2$ and at most $u-i$ packages to $S_1$. We can use this in the following way -- we construct a set of $u+1$ Cover-Free Families $\set{\CFFLetter_i}_{i=0}^u$ such that $\CFFLetter_i$ is a \parameterizedCFF{d+c+(u-i)+1}{i+2}. We then set the tests of the algorithm to be $\mathcal T\triangleq \bigcup_{i=0}^u \CFFLetter_i$. Similar arguments to those of lemmas~\ref{lem:defectIdentificationWithConflicts},~\ref{lem:learningConflictsWithUnknownDependencies}, and~\ref{lem:learningDependencies} show that we still have witnesses for every non-defective package, every non-conflicting pair of packages and every non-dependent pair. The actual $\CFFLetter_i$ that the relevant test belongs to is the number of prerequisites of the package $p$ in the case of defectiveness and dependency-learning, and the total number of prerequisites of either $p$ or $q$ in the case of determining whether the two conflict.
\end{proof}
\fi
\ifdefined\TACASVersion
\noindent In the full version~\cite{fullVersion} we discuss how adaptivity improves the complexity further.
\else
\textbf{Adding Adaptiveness:} As evident by the proofs of
lemmas~\ref{lem:defectIdentificationWithConflicts},
\ref{lem:learningConflictsWithUnknownDependencies}, and~\ref{lem:learningDependencies}, we can use a \parameterizedCFF{u+c+d}{u+1} to learn all defects, a \parameterizedCFF{u+c+d}{u+2} to learn the conflicts and a \parameterizedCFF{u+c+d+1}{u+1} to learn the dependencies.
Thus, we can start by learning all defects using \parameterizedCFFSize{u+c+d}{u+1} queries 
for learning just the defects and then proceed with \parameterizedCFFSize{u+c+1}{u+2} tests 
to determine the conflicts and dependencies. This allows us to shave a $d$ factor in the dependency on the number of defects. That can also be integrated with the union of \CFFs{} argument presented in Theorem~\ref{thm:learnAllLBWithConflictsUB} to \mbox{reduce the complexity further.}
\fi

\ignore{	
	The instance, depicted 
	in~\ref{subfig:input} consists of the input that 
	contains $P=\set{p_0,\ldots,p_9}$,
	$K=\set{(p_0,p_2),(p_2,p_4),(p_4,p_6),(p_4,p_5),(p_6,p_5),(p_6,p_7)}$, and
	$U=\set{(p_1,p_3),(p_4,p_0),(p_6,p_8),(p_8,p_6)}$ and $C=\set{\set{p_1,p_2}}$. 
	The \learnAll{} output, as shown in~\ref{subfig:output}, contains the set of defects 
	and the strongly connected components of the actual 
	dependency graph (including $U$) along with the full 
	specification of the dependencies and conflicts. Namely, 
	the output is the set of defects $D=\set{p_7,p_9}$ and the 
	graph $G=(V,E,A)$ where 
	$V=\set{\set{p_0,p_2,p_4},\set{p_1},\set{p_3},\set{p_6,p_8},\set{p_3},\set{p_7},
	\set{p_9}}, E=\set{\set{\set{p_0,p_2,p_4},\set{p_1}}},\break A=\set{(\set{p_0,p_2,p_4},
	\set{p_6,p_8}),(\set{p_0,p_2,p_4},\set{p_5}),(\set{p_6,p_8},\set{p_5}),(\set{p_1},
	\set{p_3})}$.
	}

\section{Related Work}
\label{sec:related}

Our work is related to the problem of software upgrades that has been widely
studied by the 
community~\cite{Mancinelli2006,DiCosmo2010,burrows,Trezentos,Tucker}. 
As many open source products such as Debian and Ubuntu operating systems are built from packages, 
some practical solutions for installing these products have been developed in the past. 
These solutions try to find a large subset of packages that are installable together. Most of them either use 
SAT solvers or pseudo-boolean
optimizations~\cite{Mancinelli2006,DiCosmo2010,Tucker} or develop greedy
algorithms~\cite{burrows} to derive a solution to that problem, i.e., find an
installable subset of packages that need to be installed (or upgraded). 
The Mancoosi~\cite{mancoosi} EU research project has been involved in solving
the problem of open source packages distributions and organized competitions for finding SAT solver based
solutions that maximize the number of installable package. 
\ifdefined\fullVersion
As part of the
project the authors of~\cite{milp} investigated the capabilities of MILP solvers
to handle the upgradeability problem and showed some improvements upon
pseudo-boolean~optimizations.
\else
\cite{milp} also investigated the capabilities of MILP solvers
to handle the upgradeability problem and showed some improvements upon
pseudo-boolean~optimizations.
\fi

All of the aforementioned techniques assume that the
dependencies and the conflicts are known and that the packages are indeed installable and do not 
contain defects. Our approach complements the existing techniques by learning
the real structure of the dependencies graph.

Techniques for detecting bugs in software using tests is another related research area~\cite{Martinez,aldaco,Segall2015,Yilmaz}. They are used do discover an unknown number of bugs. However, since there are dependencies between packages in our problem, the combinatorial algorithms developed for detecting bugs may not be used directly and require adjustments.

Our choice of cover-free families for tests selection is based on its
applications studied
in~\cite{Bshouty2017,hwang1987non,BERGER2001518,BUSH1984335,Chor1994,Yamada2012}. Hwang et. al. 
were the first to define cover-free families for non-adaptive group
testing~\cite{hwang1987non}. Their work has been followed by others that were using CFFs for group testing~\cite{BERGER2001518,BUSH1984335}.
\ifdefined\fullVersion

Researchers also used cover-free families for finding an $r$-simple
$k$-path~\cite{Bshouty2017} in a graph and solving
cryptographic problems~\cite{Chor1994,Yamada2012}.\fi Our approach leverages the
efficiency achieved by other researchers in computing
CFFs~\cite{CFFLB,NewestCFFLB,NewCFFLB}, thus providing efficient algorithms for
learning the entire relations graph representing packages in a software system.

\section{Discussion}

\label{sec:discussion}
\newcommand{\specialcell}[2][c]{%
	\begin{tabular}[#1]{@{}l@{}}#2\end{tabular}}
\ifdefined\TACASVersion
\begin{table}[t]
	\centering
	
	\scriptsize
	\else
	\begin{table}[htp]
	\centering
	\small
	\fi
	\begin{tabular}{|l|c|l|l|c|}
		\hline
		\textbf{Scenario}& \textbf{Adpt}& \textbf{Problem}& \textbf{Query Complexity Bounds} &  \textbf{Thm} \#		 \\
		\hline\hline

		\multirow{3}{*}{\specialcell[l]{\noindent At most
		$r$ root\\ 
		defects and no unknown\\ 
		dependencies or conflicts}} & \multirow{3}{*}{\cmark} &
	{\sc Full}&
		Lower: $\ceil{\log_2\parentheses{\sum_{i=0}^{r}{n\choose
		i}}}=\Omega(r\logp{\frac{n}{r}})$ & \ref{thm:GTLB} \\&& 	{\sc Learning} &&\\
		&&& Upper: $r-1 + \ceil{\log_2\parentheses{\sum_{i=0}^{r}{n\choose
		i}}}$  &
		\ref{thm:knownDependenciesUpperBound}\\\hline

		\multirow{3}{*}{\specialcell[l]{\noindent At most $1$
		unknown \\
		dependency and \\ $1$ root defect}}& \multirow{3}{*}{\cmark} & {\sc Maximal} & \multirow{3}{*}{Lower:
		$n=\Omega(n)$} & \multirow{3}{*}{\ref{thm:oneRootOneUnknownDep}}\\
		&&{\sc Sub-}&  &  \\&& {\sc repository} &&\\\hline

		\multirow{5}{*}{\specialcell[l]{\noindent At most $d$
		defects and $u$\\ unknown dependencies,\\ no conflicts}}& \multirow{3}{*}{\xmark} &&
		Lower: $\parameterizedCFFSize{d-1}{u+1}=\Omega^*(\log n)$ &
		\ref{thm:unknownDependenciesLB}   \\&&{\sc Maximal}&&\\& &{\sc Sub-}&
		Upper: $\parameterizedCFFSize{d}{u+1}=O^*(\log n)$ &
		\ref{thm:unknownDependenciesUB}\\\hhline{~-~--}
		 & \multirow{2}{*}{\cmark}& {\sc repository} &
		\multirow{2}{*}{Lower: $\Omega\parentheses{{u+d\choose d}+d\logp{n\over d}}=\Omega^*(\log n)$} &
		\multirow{2}{*}{\ref{thm:adaptiveMaxInstallation}}\\&&&&\\\hline	

		\multirow{3}{*}{\specialcell[l]{\noindent At most $1$ conflict and \\ no defects or \\unknown dependencies}} & \multirow{3}{*}{\xmark}&{\sc Full} &
		\multirow{3}{*}{Lower: $n-1=\Omega(n)$} &
		\multirow{3}{*}{\ref{thm:learningAStrongConflict}}\\&& {\sc Learning}&& \\
		&&&& \\\hline	
		
		\ifdefined\fullVersion
		
		\multirow{5}{*}{\specialcell[l]{\noindent At most
				$c$ weak\\ conflicts, $d$ defects and\\ no unknown dependencies}}& \multirow{3}{*}{\xmark} &
		&
		Lower: $\parameterizedCFFSize{d+c-1}{2}=\Omega^*(\log n)$ & \ref{thm:learnAllLBWithConflictsKnownDependenciesLB} \\&&
		{\sc Full}&&\\
		&& {\sc Learning} & Upper: $\parameterizedCFFSize{d+c}{2}=O^*(\log n)$  &
		\ref{thm:learnAllLBWithConflictsUBKnownDependencies}\\\hhline{~-~--}	
		& \multirow{2}{*}{\cmark}&&
		\multirow{2}{*}{Upper: $\parameterizedCFFSize{d+c}{1}+\parameterizedCFFSize{c}{2}=O^*(\log n)$} &
		\multirow{2}{*}{\ref{thm:adaptiveConflictsWithKnownDependencies}}\\
		&&&& \\\hline	
		
		\else
		
		\multirow{3}{*}{\specialcell[l]{\noindent At most
				$c$ weak\\ conflicts, $d$ defects and\\ no unknown dependencies}}& \multirow{3}{*}{\xmark} &
		{\sc Full}
		&
		\multirow{3}{*}{Lower: $\parameterizedCFFSize{d+c-1}{2}=\Omega^*(\log n)$} & \multirow{3}{*}{\ref{thm:learnAllLBWithConflictsKnownDependenciesLB}} \\&&{\sc Learning}&&\\
		&&&& \\\hline	
	
		\fi
		
		\multirow{3}{*}{\specialcell[l]{\noindent At most $1$ \\unknown
		dependency  and\\ no conflicts or defects}} & \multirow{3}{*}{\xmark}&{\sc Full} &
	\multirow{3}{*}{Lower: $n-1=\Omega(n)$} &
	\multirow{3}{*}{\ref{thm:learningADependency}}\\&&{\sc Learning}&&\\
	&&&& \\\hline				
		\multirow{8}{*}{\specialcell[l]{\noindent At most $u$ unknown\\
				dependencies, $c$ conflicts,\\ and $d$  defects.\\All dependencies are \\ unknown}} & \multirow{8}{*}{\xmark}&{\sc Maximal} &
		\multirow{3}{*}{Lower:
		$\parameterizedCFFSize{d+\ceil{c/2}-1}{u+1}=\Omega^*(\log n)$} & \multirow{3}{*}{\ref{thm:maxLowerBoundWithConflicts}} \\&&{\sc
		Sub-}&& \\&&{\sc
		repository}&& \\\hhline{~~---}		
		&
		& & Lower:
		$\parameterizedCFFSize{d+c-1}{u+2}=\Omega^*(\log n)$ &\ref{thm:learnAllLBWithConflictsV2} \\&&{\sc Full}&&\\
		&&{\sc Learning}&
		Upper: 
		$\!\begin{aligned} 
		&\sum_{i=0}^{u}\parameterizedCFFSize{d+c+(u-i)+1}{i+2} \\    
		&< \parameterizedCFFSize{d+c+u+1}{u+2}=O^*(\log n) \end{aligned}$
		& 
		\ref{thm:learnAllLBWithConflictsUB} \\\hline \end{tabular}\smallskip
	\normalsize
	\caption{Summary of our results learning undocumented constraints. Here, $\CFFSizeLetter_{\CFFParams}$ is the size of a \CFF{} as described in Section~\ref{sec:CFFs}. Recall that every algorithm for \learnAll{} also solves \maxProblem{} and every lower bound for \maxProblem{} holds for \learnAll{}.}
	\label{tbl:summary}
	\vspace{-.5cm}
\end{table} 

In this paper, we formalized a stylized model that allows us to reason about the
complexity of learning undocumented software constraints in a given repository.
We presented algorithms for four cases: where the entire dependencies structure
is known, and we are interested in finding the defects; where we have up to $u$
unknown dependencies, and we wish to find all the defects and the dependencies;
where we have up to $c$ unknown conflicts, and we wish to find all the defects
and conflicts; and where up to $u$ dependencies and up to $c$ conflicts are
unknown.
We proved lower and upper bounds on the complexity of both adaptive algorithms
for solving the problems of \learnAll{} and \maxProblem{}.
Table~\ref{tbl:summary}  gives the summary of our results.

\ifdefined\fullVersion
The first column in the table describes the scenario in which the algorithm is
applied, the second states whether the algorithm is adaptive, while the third
one states which problem is being solved.
The fourth column states the query complexity bounds. There, we present the
lower bound for all of the scenarios and problem variants. For the scenarios
where the lower bound for query complexity was proved to be linear, we do not
state the upper bound as it is at least linear. For simplicity of comparison
between the different cases, we introduce the following notation.  We say that
$O^*(f(n))$  is the asymptotic complexity for constant-many unknowns.
Specifically, a function $g(n,u,c,d)$ is in $O^*(f(n))$  if it is in $O(f(n)
\cdot h(u,c,d))$ for some function $h(u,c,d)=2^{O(u+c+d)}$. The upper and lower
bounds for query complexity can then be presented using $O^*$ and $\Omega^*$
notations for small $d$, $c$, and $u$. The last column gives a reference to the
theorem that proves the result.
We note that the problems are easily solvable in $O^*(\mbox{poly}(n))$, as all
constraints are local. For example, to determine whether $p$ has an unknown
dependency on $q$ we need to test $p$ with all its prerequisites and exclude $q$
and its dependents. As there are $O(n^2)$ vertex pairs to test, this is doable
in the mentioned query complexity. Alas, as $n$ is large and testing is costly,
we wish to determine which scenarios are learnable in logarithmic~complexity.
\else
The table describes the scenario in which the algorithm is
applied, whether the algorithm is adaptive, and which problem is being solved.
Then it gives the query complexity bounds. The
lower bound is shown for all cases. For the scenarios
where the lower bound for query complexity was proved to be linear, we do not
state the upper bound as it is at least linear. For simplicity of comparison
between the cases, we introduce the following notation.  We say that
$O^*(f(n))$  is the asymptotic complexity for constant-many unknowns.
Specifically, a function $g(n,u,c,d)$ is in $O^*(f(n))$  if it is in $O(f(n)
\cdot h(u,c,d))$ for some function $h(u,c,d)=2^{O(u+c+d)}$. The upper and lower
bounds for query complexity can then be presented using $O^*$ and $\Omega^*$
notations for small $d$, $c$, and $u$. 
\fi

Our algorithms assume that there are no \emph{optional dependencies} in the
package repository for which we need to find a solution.  This means that if a
package \textit{a} depends on \textit{b} then any working installation that
includes \textit{a} must include \textit{b} as well. However, for some package repositories, 
it might be the case that \textit{a} depends
on either \textit{b} \textbf{or} \textit{c}. \ifdefined\fullVersion Thus a working installation could
be $\{a, b\}$ or $\{a, c\}$, while $\{a\}$ will fail. \fi In the future, we plan to investigate how our algorithms are affected by this definition of dependencies.

Another fundamental issue is setting the bounds on the number of unknown constraints. That is, we may not know in advance how many defects, hidden dependencies, or unknown conflicts exists in a repository. 
A possible workaround would be devising an \emph{anytime algorithm} -- an algorithm that has no predefined stopping condition but can generate output on demand. The natural operation for such an algorithm would be starting from small values for $d,c,u$ and incrementing them one at a time. The query complexity dependency on the parameters is such that the overall number of queries would not be much larger than if we had known the parameters to begin with.
\ignore{
\\1
\\2
\\3
\\4
\\5
\\6
\\7
\\8
\\9
\\10
\\11
\\12
\\13
\\14
\\15
}
\ifdefined\fullVersion
If the parameters, at the time of output, are not large enough to capture the actual number of unknown constraints then either the algorithms will fail to find a working installation or would mistake some packages as defective while in practice they are functional but with many unknown constraints.
\fi

\ifdefined\TACASVersion
\newpage
\fi

\bibliographystyle{splncs03}
\bibliography{DependenciedAndConflicts}

\begin{thebibliography}{10}
\providecommand{\url}[1]{\texttt{#1}}
\providecommand{\urlprefix}{URL }

\bibitem{burrows}
Modelling and resolving software dependencies.
  https://people.debian.org/~dburrows/model.pdf (2005)

\bibitem{aptget}
{AptGet -Package management with APT}. https://help.ubuntu.com/community/AptGet
  (2017)

\bibitem{aptitude}
Aptitude package manager. https://wiki.debian.org/Aptitude (2017)

\bibitem{cupt}
Cupt package manager. https://wiki.debian.org/Cupt (2017)

\bibitem{mancoosi}
The mancoosi project. http://mancoosi.org/papers/ (2017)

\bibitem{smart}
Smart package manager. https://labix.org/smart/ (2017)

\bibitem{NewestCFFLB}
Abdi, A.Z., Bshouty, N.H.: Lower bounds for cover-free families. The Electronic
  J. of Combinatorics  (2016)

\bibitem{aldaco}
Aldaco, A., Colbourn, C., Syrotiuk, V.: Locating arrays: A new experimental
  design for screening complex engineered systems. In: Operating Systems Review
  (ACM). vol.~49, pp. 31--40. Association for Computing Machinery (1 2015)

\bibitem{BERGER2001518}
Berger, T., Levenshtein, V.: Application of cover-free codes and combinatorial
  designs to two-stage testing. Electronic Notes in Discrete Mathematics  6,
  518 -- 527 (2001)

\bibitem{Bshouty2017}
Bshouty, N.H., Gabizon, A.: Almost optimal cover-free families. In: CIAC. pp.
  140--151 (2017)

\bibitem{BUSH1984335}
Bush, K., Federer, W., Pesotan, H., Raghavarao, D.: New combinatorial designs
  and their applications to group testing. Journal of Statistical Planning and
  Inference  10(3),  335 -- 343 (1984)

\bibitem{Chor1994}
Chor, B., Fiat, A., Naor, M.: Tracing Traitors, pp. 257--270. Springer Berlin
  Heidelberg (1994)

\bibitem{DiCosmo2010}
Di~Cosmo, R., Boender, J.: Using strong conflicts to detect quality issues in
  component-based complex systems. In: Proceedings of the 3rd India Software
  Engineering Conference. ISEC, ACM (2010)

\bibitem{Dorfman}
Dorfman, R.: The detection of defective members of large populations. The
  Annals of Mathematical Statistics  (1943)

\bibitem{NewCFFLB}
Hajiabolhassan, H., Moazami, F.: Some new bounds for cover-free families
  through biclique covers. Discrete Math.  (2012)

\bibitem{Hwang}
Hwang, F.K.: A method for detecting all defective members in a population by
  group testing. Journal of the American Statistical Association  67(339)
  (1972)

\bibitem{hwang1987non}
Hwang, F., S{\'o}s, V.: Non-adaptive hypergeometric group testing. Studia Sci.
  Math. Hungar  22 (1987)

\bibitem{Kahn}
Kahn, A.B.: Topological sorting of large networks. Commun. ACM  5 (Nov 1962)

\bibitem{li1962sequential}
Li, C.H.: A sequential method for screening experimental variables. Journal of
  the American Statistical Association  (1962)

\bibitem{Mancinelli2006}
Mancinelli, F., Boender, J., di~Cosmo, R., Vouillon, J., Durak, B., Leroy, X.,
  Treinen, R.: Managing the complexity of large free and open source
  package-based software distributions. In: IEEE/ACM ASE (2006)

\bibitem{Martinez}
Mart\'{\i}nez, C., Moura, L., Panario, D., Stevens, B.: Locating errors using
  elas, covering arrays, and adaptive testing algorithms. SIAM J. Discret.
  Math.  23(4),  1776--1799 (Dec 2009)

\bibitem{milp}
Michel, C., Rueher, M.: Handling software upgradeability problems with {MILP}
  solvers. In: LoCoCo (2010)

\bibitem{Segall2015}
Segall, I., Tzoref-Brill, R.: Feedback-driven combinatorial test design and
  execution. In: Proceedings of the 8th ACM International Systems and Storage
  Conference. pp. 12:1--12:6. SYSTOR '15, ACM, New York, NY, USA (2015),
  \url{http://doi.acm.org/10.1145/2757667.2757677}

\bibitem{CFFLB}
Stinson, D.R., Wei, R., Zhu, L.: Some new bounds for cover-free families. J. of
  Comb. Theory  (2000)

\bibitem{Trezentos}
Trezentos, P., Lynce, I., Oliveira, A.L.: Apt-pbo: Solving the software
  dependency problem using pseudo-boolean optimization. In: IEEE/ACM ASE. pp.
  427--436. ACM, New York, NY, USA (2010)

\bibitem{Tucker}
Tucker, C., Shuffelton, D., Jhala, R., Lerner, S.: Opium: Optimal package
  install/uninstall manager. In: Proceedings of the 29th International
  Conference on Software Engineering. pp. 178--188. ICSE '07, IEEE Computer
  Society, Washington, DC, USA (2007),
  \url{https://doi.org/10.1109/ICSE.2007.59}

\bibitem{Yamada2012}
Yamada, S., Hanaoka, G., Kunihiro, N.: Two-Dimensional Representation of Cover
  Free Families and Its Applications: Short Signatures and More, pp. 260--277.
  Springer Berlin Heidelberg (2012)

\bibitem{Yilmaz}
Yilmaz, C., Dumlu, E., Cohen, M.B., Porter, A.: Reducing masking effects in
  combinatorialinteraction testing: A feedback drivenadaptive approach. IEEE
  Transactions on Software Engineering  40(1),  43--66 (Jan 2014)

\end{thebibliography}

\end{document}